\documentclass[11pt]{article}
\usepackage{amssymb, amsmath, amsthm}

\usepackage{amsfonts,mathtools,mathrsfs,mathtools,color}
\usepackage[colorlinks]{hyperref}
\usepackage{graphicx,subfigure}

\makeatletter
\newtheorem*{rep@theorem}{\rep@title}
\newcommand{\newreptheorem}[2]{%
\newenvironment{rep#1}[1]{%
 \def\rep@title{#2 \ref{##1}}%
 \begin{rep@theorem}}%
 {\end{rep@theorem}}}
\makeatother

\newreptheorem{theorem}{Theorem}
\newtheorem{theorem}{Theorem}

\newtheorem*{theorem*}{Theorem}
\newtheorem{lemma}{Lemma}
\newtheorem{corollary}{Corollary}
\newtheorem{example}{Example}
\newtheorem{assumption}{Assumption}
\newtheorem{remark}{Remark}

\newcommand{\R}{\mathbb{R}} 
\newcommand{\N}{\mathcal{N}}
\newcommand{\E}{\mathbb{E}}

\renewcommand{\part}[2]{\frac{\partial #1}{\partial #2}}

\newcommand{\Var}{\mathrm{Var}}
\newcommand{\Tr}{\mathrm{Tr}}

\newcommand{\step}{\eta}
\newcommand{\err}{\delta}

\newcommand{\Cov}{\mathrm{Cov}}
\newcommand{\op}{\mathrm{op}}
\newcommand{\HS}{\mathrm{HS}}

\title{Rapid Convergence of the Unadjusted Langevin Algorithm: Isoperimetry Suffices}

\author{Santosh S.\ Vempala\thanks{Georgia Institute of Technology, College of Computing. Email: \texttt{vempala@gatech.edu}. Supported in part by NSF awards CCF-1563838, CCF-1717349 and DMS-1839323. }
\and 
Andre Wibisono\thanks{Yale University, Department of Computer Science. Email: \texttt{andre.wibisono@yale.edu}}}

\begin{document}

\maketitle

\begin{abstract}
We study the Unadjusted Langevin Algorithm (ULA) for sampling from a probability distribution $\nu = e^{-f}$ on $\R^n$.
We prove a convergence guarantee in Kullback-Leibler (KL) divergence assuming $\nu$ satisfies a log-Sobolev inequality and the Hessian of $f$ is bounded. 
Notably, we do not assume convexity or bounds on higher derivatives.
We prove convergence guarantees in R\'enyi divergence of order $q > 1$ assuming the limit of ULA satisfies isoperimetry, namely either the log-Sobolev or Poincar\'e inequality.
We also prove a bound on the bias of the limiting distribution of ULA assuming third-order smoothness of $f$, without requiring isoperimetry.
\end{abstract}

\setcounter{tocdepth}{2}
\tableofcontents
\newpage

\section{Introduction}

Sampling is a fundamental algorithmic task.
Many applications require sampling from probability distributions in high-dimensional spaces, and in modern applications the probability distributions are complicated and non-logconcave.
While the setting of logconcave functions is well-studied, it is important to have efficient sampling algorithms with good convergence guarantees beyond the logconcavity assumption.
There is a close interplay between sampling and optimization, either via optimization as a limit of sampling (annealing)~\cite{HS88,RRT17}, or via sampling as optimization in the space of distributions~\cite{JKO98,W18}.
Motivated by the widespread use of non-convex optimization and sampling, there is resurgent interest in understanding non-logconcave sampling.

In this paper we study a simple algorithm, the Unadjusted Langevin Algorithm (ULA), for sampling from a target probability distribution $\nu = e^{-f}$ on $\R^n$.
ULA requires oracle access to the gradient $\nabla f$ of the log density $f = -\log \nu$.
In particular, ULA does not require knowledge of $f$, which makes it applicable in practice where we often only know $\nu$ up to a normalizing constant.

As the step size $\step \to 0$, ULA recovers the Langevin dynamics, which is a continuous-time stochastic process in $\R^n$ that converges to $\nu$.
We recall the optimization interpretation of the Langevin dynamics for sampling as the gradient flow of the Kullback-Leibler (KL) divergence with respect to $\nu$ in the space of probability distributions with the Wasserstein metric~\cite{JKO98}.
When $\nu$ is strongly logconcave, the KL divergence is a strongly convex objective function, so the Langevin dynamics as gradient flow converges exponentially fast~\cite{BE85,Vil03}. From the classical theory of Markov chains and diffusion processes, there are several known conditions milder than logconcavity that are sufficient for rapid convergence {\em in continuous time}.
These include isoperimetric inequalities such as Poincar\'e inequality or log-Sobolev inequality (LSI).
Along the Langevin dynamics in continuous time, Poincar\'e inequality implies an exponential convergence rate in $L^2(\nu)$, while LSI---which is stronger---implies an exponential convergence rate in KL divergence (as well as in R\'enyi divergence).

However, in discrete time, sampling under Poincar\'e inequality or LSI is a  more challenging problem.
ULA is an inexact discretization of the Langevin dynamics, and it converges to a biased limit $\nu_\step \neq \nu$.
When $\nu$ is strongly logconcave and smooth, it is known how to control the bias and prove a convergence guarantee on KL divergence along ULA; see for example~\cite{CB18,D17,DK19,DMM18}.
When $\nu$ is strongly logconcave, there are many other sampling algorithms with provable rapid convergence; these include the ball walk and hit-and-run \cite{KLS97, LV07, LV06, LV06b} (which give truly polynomial algorithms), various discretizations of the overdamped or underdamped Langevin dynamics~\cite{D17,DK19,DMM18,B18,DCWY18} (which have polynomial dependencies on smoothness parameters but low dependence on dimension), and more sophisticated methods such as the Hamiltonian Monte Carlo~\cite{MS17,MV18,DMS17,LS18,CV19}.
It is of great interest to extend these results to non-logconcave densities $\nu$, where existing results require strong assumptions with bounds that grow exponentially with the dimension or other parameters~\cite{AK91,ChengLangevin18,Ma19,MV19}.
There are also recent works that analyze convergence of sampling using various techniques such as reflection coupling~\cite{EGZ19}, kernel methods~\cite{GM17}, and higher-order integrators~\cite{LWME19}, albeit still under some strong conditions such as distant dissipativity, which is similar to strong logconcavity outside a bounded domain.

In this paper we study the convergence along ULA under minimal (and necessary) isoperimetric assumptions, namely, LSI and Poincar\'e inequality. 
These are sufficient for fast convergence in continuous time; moreover, in the case of logconcave distribution, the log-Sobolev and Poincar\'e constants can be bounded and lead to convergence guarantees for efficient sampling in discrete time. 
However, do they suffice on their own without the assumption of logconcavity?  

We note that LSI and Poincar\'e inequality apply to a wider class of measures than logconcave distributions.
In particular, LSI and Poincar\'e inequality are preserved under bounded perturbation and Lipschitz mapping (see Lemma~\ref{Lem:LSILipschitz} and Lemma~\ref{Lem:PLipschitz}), whereas logconcavity would be destroyed.
Given these properties, it is easy to exhibit examples of non-logconcave distributions satisfying LSI or Poincar\'e inequality.
For example, we can take a small perturbation of a convex body to make it nonconvex but still satisfies isoperimetry; then the uniform probability distribution on the body (or a smooth approximation of it) is not logconcave but satisfies LSI or Poincar\'e inequality.
Similarly, we can start with a strongly logconcave distribution such as a Gaussian, and subtract some small Gaussians from it; then the resulting (normalized) probability distribution is not logconcave, but it still satisfies LSI or Poincar\'e inequality as long as the Gaussians we subtract are small enough.
See Figure~\ref{Fig:Examples} for an illustration.

 \begin{figure}[h!t!]
 \centering
 \begin{subfigure}
   \centering
    \includegraphics[width=0.23\textwidth]{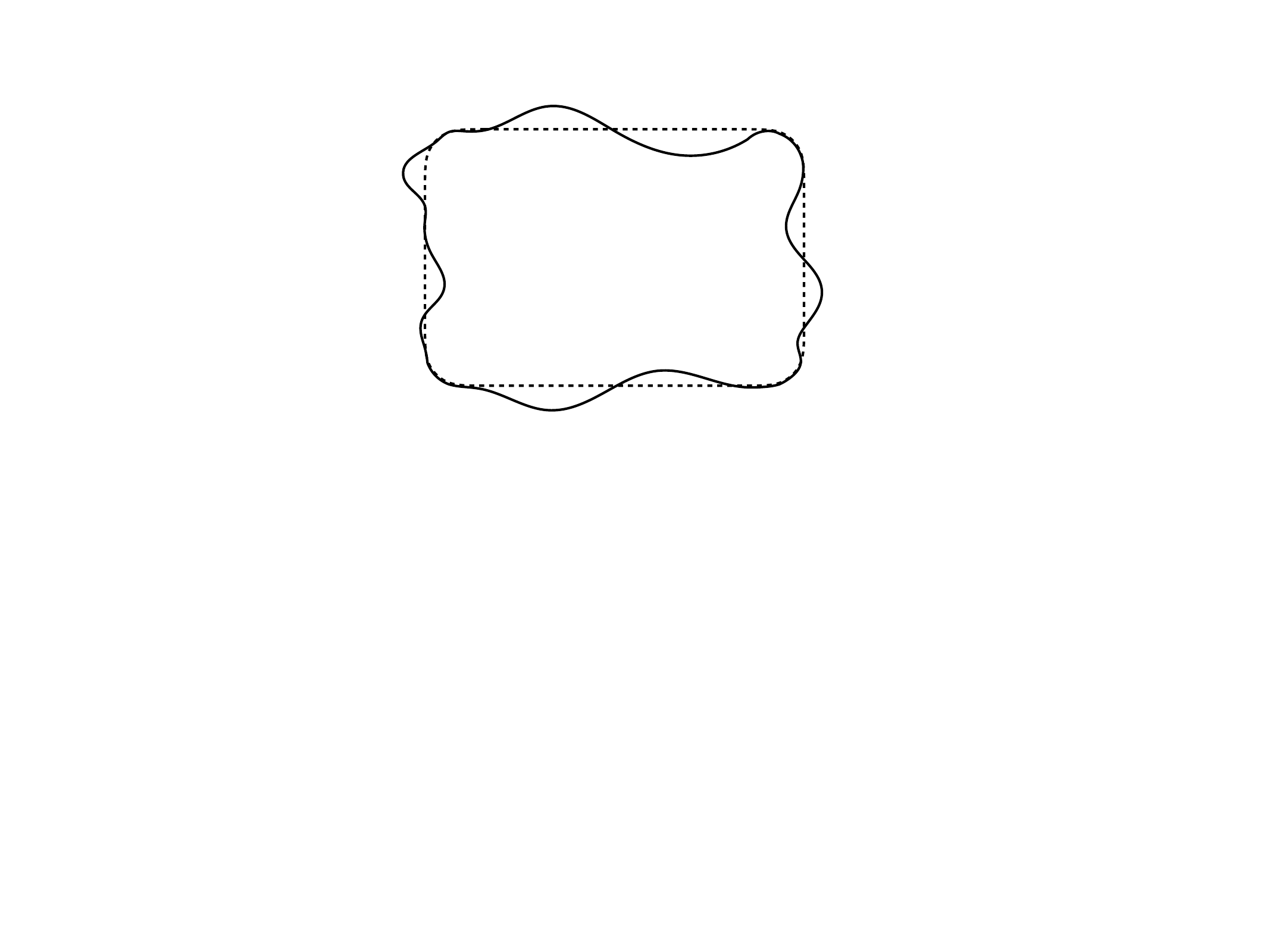}
 \end{subfigure}
  ~~~~~~
 \begin{subfigure}
   \centering
    \includegraphics[width=0.6\textwidth]{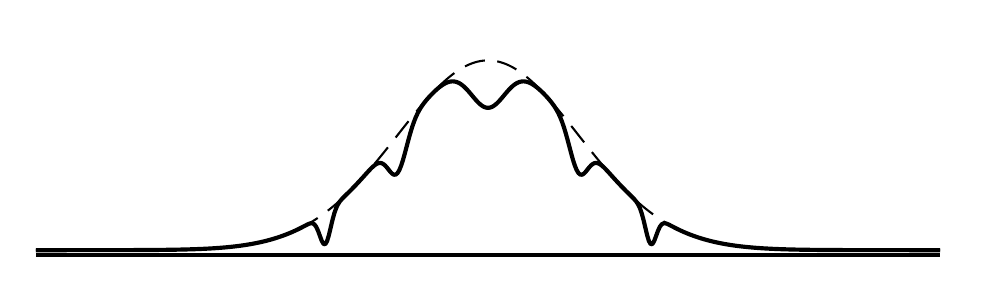}
 \end{subfigure}
 \caption{Illustrations of non-logconcave distributions satisfying LSI or Poincar\'e inequality: the uniform distribution on a nonconvex set (left), and a small perturbation of a logconcave distribution, e.g., Gaussian (right).}
 \label{Fig:Examples}
 \end{figure}

We measure the mode of convergence using KL divergence and R\'enyi divergence of order $q \ge 1$, which is stronger. 
Our first main result says that the only further assumption we need is smoothness, i.e., the gradient of $f$ is Lipschitz (see Section~\ref{Sec:KL-ULA-LSI}). 
Here $H_\nu(\rho)$ is the KL divergence between $\rho$ and $\nu$.
We say that $\nu = e^{-f}$ is $L$-smooth if $\nabla f$ is $L$-Lipschitz, or equivalently, $-LI \preceq \nabla^2 f(x) \preceq LI$ for all $x \in \R^n$ .

\begin{theorem}\label{Thm:Main}
Assume $\nu = e^{-f}$ satisfies log-Sobolev inequality with constant $\alpha > 0$ and is $L$-smooth.
ULA with step size $0 < \step \le \frac{\alpha}{4L^2}$ satisfies
\begin{align*}
H_\nu(\rho_k) \le e^{-\alpha \step k} H_\nu(\rho_0) + \frac{8 \step n L^2}{\alpha}.
\end{align*}
In particular, for any $0 < \err < 4n$, ULA with step size $\step \le \frac{\alpha\err}{16L^2n}$ reaches error $H_\nu(\rho_k) \le \err$ after $k \ge \frac{1}{\alpha \step} \log \frac{2 H_\nu(\rho_0)}{\err}$ iterations.
\end{theorem}

For example, if we start with a Gaussian $\rho_0 = \N(x^\ast, \frac{1}{L}I)$ where $x^\ast$ is a stationary point of $f$ (which we can find, e.g., via gradient descent), then $H_\nu(\rho_0) = \tilde O(n)$ (see Lemma~\ref{Lem:KLEstimate}), 
and Theorem~\ref{Thm:Main} gives an iteration complexity of 
$k = \tilde \Theta\left(\frac{L^2 n}{\alpha^2 \err}\right)$
to achieve $H_\nu(\rho_k) \le \err$ using ULA with step size $\step = \Theta(\frac{\alpha \err}{L^2 n})$. 

The result above matches previous known bounds for ULA when $\nu$ is strongly logconcave~\cite{CB18,D17,DK19,DMM18}.
Our result complements the work of Ma et al.~\cite{Ma19} who study the underdamped version of the Langevin dynamics under LSI and show an iteration complexity for the discrete-time algorithm that has better dependence on the dimension ($\sqrt{\frac{n}{\err}}$ in place of $\frac{n}{\err}$ above for ULA), but under an additional smoothness assumption ($f$ has bounded third derivatives) and with higher polynomial dependence on other parameters. Our result also complements the work of Mangoubi and Vishnoi~\cite{MV19} who study the Metropolis-adjusted version of ULA (MALA) for non-logconcave $\nu$ and show a $\log(\frac{1}{\err})$ iteration complexity from a warm start, under the additional assumption that $f$ has bounded third and fourth derivatives in an appropriate $\infty$-norm.

We note that in general some isoperimetry condition is needed for rapid mixing of Markov chains (such as the Langevin dynamics and ULA), otherwise there are bad regions in the state space from which the chains take arbitrarily long to escape.
Smoothness or bounded Hessian is a common assumption that seems to be needed for the analysis of discrete-time algorithms (such as gradient descent or ULA above).

In the second part of this paper, we study the convergence of R\'enyi divergence of order $q > 1$ along ULA.
R\'enyi divergence is a family of generalizations of KL divergence~\cite{R61,VH14,BCG19}, which becomes stronger as the order $q$ increases.
There are physical and operational interpretations of R\'enyi divergence~\cite{H06,B11}.
R\'enyi divergence has been useful in many applications, including for the exponential mechanism in differential privacy~\cite{DR16,ACGM16,BS16,M17}, 
lattice-based cryptography~\cite{BLRS18},
information-theoretic encryption~\cite{IS13}, 
variational inference~\cite{LT16},
machine learning~\cite{HHK03,MMR09},
information theory and statistics~\cite{C95,MPV00},
and black hole physics~\cite{D16}.

Our second main result proves a convergence bound for the R\'enyi divergence of order $q > 1$. While this is a stronger measure of convergence than KL divergence, the situation here is more complicated. First, we can only hope to converge to the target for finite $q$ for any step-size $\step$ (as we illustrate with an example). Second, it is unclear how to bound the R\'enyi divergence between the biased limit $\nu_\step$ and $\nu$. We first show the convergence of R\'enyi divergence along Langevin dynamics in continuous time under LSI; see Theorem~\ref{Thm:RenyiLD} in Section~\ref{Sec:RenyiLDRate}. Here $R_{q,\nu}(\rho)$ is the R\'enyi divergence of order $q$ between $\rho$ and $\nu$. 

\begin{theorem}\label{Thm:RenyiLD}
Suppose $\nu$ satisfies LSI with constant $\alpha > 0$.
Let $q \ge 1$.
Along the Langevin dynamics,
\begin{align*}
R_{q,\nu}(\rho_t) \le e^{-\frac{2\alpha t}{q}} R_{q,\nu}(\rho_0).
\end{align*}
\end{theorem}

We also have the following convergence of R\'enyi divergence along Langevin dynamics under Poincar\'e inequality; see Theorem~\ref{Thm:RenyiLP} in Section~\ref{Sec:PoincareLangevin}.

\begin{theorem}\label{Thm:RenyiLP}
Suppose $\nu$ satisfies Poincar\'e inequality with constant $\alpha > 0$.
Let $q \ge 2$.
Along the Langevin dynamics,
\begin{align*}
R_{q,\nu}(\rho_t) \le
\begin{cases}
R_{q,\nu}(\rho_0) -\frac{2\alpha t}{q} ~~ & \text{ if } R_{q,\nu}(\rho_0) \ge 1 \text{ and as long as } R_{q,\nu}(\rho_t) \ge 1, \\
e^{-\frac{2\alpha t}{q}} R_{q,\nu}(\rho_0) ~~ & \text{ if } R_{q,\nu}(\rho_0) \le 1.
\end{cases}
\end{align*}
\end{theorem}
The reader will notice that under Poincar\'e inequality, compared to LSI, the convergence is slower in the beginning before it becomes exponential. For a reasonable starting distribution (such as a Gaussian centered at a stationary point), this leads to an extra factor of $n$ compared to the convergence under LSI.

We then turn to the discrete-time algorithm and show that ULA converges in R\'enyi divergence to the biased limit $\nu_\step$ under the assumption that $\nu_\step$ itself satisfies either LSI or Poincar\'e inequality.
We combine this with a decomposition result on R\'enyi divergence to derive a convergence guarantee in R\'enyi divergence to $\nu$; see Theorem~\ref{Thm:RenyiRate} in Section~\ref{Sec:RenyiULALSI} and Theorem~\ref{Thm:RenyiRatePoincare} in Section~\ref{Sec:RenyiRatePoincare}.

Finally, we show some properties on the biased limit of ULA.
We bound the bias in relative Fisher information assuming third-order smoothness (without isoperimetry); see Theorem~\ref{Thm:Bias}.
We also show the biased limit satisfies LSI if the original target is smooth and strongly log-concave; see Theorem~\ref{Thm:BiasLSI}.

In what follows, we review KL divergence and its properties along the Langevin dynamics in Section~\ref{Sec:Review}, and prove a convergence guarantee for KL divergence along ULA under LSI in Section~\ref{Sec:ULA}.
We provide a review of R\'enyi divergence and its properties along the Langevin dynamics in Section~\ref{Sec:Renyi}.
We then prove the convergence guarantee for R\'enyi divergence along ULA under LSI in Section~\ref{Sec:RenyiULA}, and under Poincar\'e inequality in Section~\ref{Sec:Poincare}. We show properties on the biased limit of ULA in Section~\ref{Sec:Bias}.
We provide all proofs and details in Section~\ref{Sec:Proofs}.
We conclude with a discussion in Section~\ref{Sec:Disc}, including subsequent work that used some of the analaysis techniques from this paper.

\section{Review of KL divergence along Langevin dynamics}
\label{Sec:Review}

In this section we review the definition of Kullback-Leibler (KL) divergence, log-Sobolev inequality, and the convergence of KL divergence along the Langevin dynamics in continuous time under log-Sobolev inequality.
See Appendix~\ref{App:Notation} for a review on notation.

\subsection{KL divergence}
\label{Sec:KL}

Let $\rho, \nu$ be probability distributions on $\R^n$, represented via their probability density functions with respect to the Lebesgue measure on $\R^n$.
We assume $\rho,\nu$ have full support and smooth densities.

Recall the {\bf Kullback-Leibler (KL) divergence} of $\rho$ with respect to $\nu$ is
\begin{align}\label{Eq:Hnu}
H_\nu(\rho) = \int_{\R^n} \rho(x) \log \frac{\rho(x)}{\nu(x)} \, dx.
\end{align}
KL divergence is the relative form of {\em Shannon entropy} $H(\rho) = -\int_{\R^n} \rho(x) \log \rho(x) \, dx$.
Whereas Shannon entropy can be positive or negative, KL divergence is nonnegative and minimized at $\nu$: $H_\nu(\rho) \ge 0$ for all $\rho$, and $H_\nu(\rho) = 0$ if and only if $\rho = \nu$.
Therefore, KL divergence serves as a measure of (albeit asymmetric) ``distance'' of a probability distribution $\rho$ from a base distribution $\nu$.
KL divergence is a relatively strong measure of distance; for example, Pinsker's inequality implies that KL divergence controls total variation distance.
Furthermore, under log-Sobolev (or Talagrand) inequality, KL divergence also controls the quadratic Wasserstein $W_2$ distance, as we review below.

We say $\nu = e^{-f}$ is {\bf $L$-smooth} if $f$ has bounded Hessian: $-LI \preceq \nabla^2 f(x) \preceq LI$ for all $x \in \R^n$.

\begin{lemma}\label{Lem:KLEstimate}
Suppose $\nu = e^{-f}$ is $L$-smooth.
Let $\rho = \N(x^\ast, \frac{1}{L} I)$ where $x^\ast$ is a stationary point of $f$.
Then $H_\nu(\rho) \le f(x^\ast) + \frac{n}{2} \log \frac{L}{2\pi}$.
\end{lemma}

We provide the proof of Lemma~\ref{Lem:KLEstimate} in Section~\ref{App:ProofLemKLEstimate}.

\subsection{Log-Sobolev inequality}

Recall we say $\nu$ satisfies the {\bf log-Sobolev inequality (LSI)} with a constant $\alpha > 0$ if for all smooth function $g \colon \R^n \to \R$ with $\E_\nu[g^2] < \infty$,
\begin{align}\label{Eq:LSI}
\E_\nu[g^2 \log g^2] - \E_\nu[g^2] \log \E_\nu[g^2] \le \frac{2}{\alpha} \E_\nu[\|\nabla g\|^2].
\end{align}
Recall the {\bf relative Fisher information} of $\rho$ with respect to $\nu$ is
\begin{align}\label{Eq:Jnu}
J_\nu(\rho) = \int_{\R^n} \rho(x) \left\| \nabla \log \frac{\rho(x)}{\nu(x)} \right\|^2 dx.
\end{align}
LSI is equivalent to the following relation between KL divergence and Fisher information for all $\rho$:
\begin{align}\label{Eq:LSI-KL}
H_\nu(\rho) \le \frac{1}{2\alpha} J_\nu(\rho).
\end{align}
Indeed, to obtain~\eqref{Eq:LSI-KL} we choose $g^2 = \frac{\rho}{\nu}$ in~\eqref{Eq:LSI}; conversely, to obtain~\eqref{Eq:LSI} we choose $\rho = \frac{g^2 \nu}{\E_\nu[g^2]}$ in~\eqref{Eq:LSI-KL}.

LSI is a strong isoperimetry statement and implies, among others, concentration of measure and sub-Gaussian tail property~\cite{L99}.
LSI was first shown by Gross~\cite{G75} for the case of Gaussian $\nu$.
It was extended by Bakry and \'Emery~\cite{BE85} to strongly log-concave $\nu$; namely, when $f = -\log \nu$ is $\alpha$-strongly convex, then $\nu$ satisfies LSI with constant $\alpha$.
However, LSI applies more generally.
For example, the classical perturbation result by  Holley and Stroock~\cite{HS87} states that LSI is stable under bounded perturbation.
Furthermore, LSI is preserved under a Lipschitz mapping.
In one dimension, there is an exact characterization of when a probability distribution on $\R$ satisfies LSI~\cite{BG99}.
Moreover, LSI satisfies a tensorization property~\cite{L99}: If $\nu_1, \nu_2$ satisfy LSI with constants $\alpha_1,\alpha_2> 0$, respectively, then $\nu_1 \otimes \nu_2$ satisfies LSI with constant $\min \{\alpha_1,\alpha_2\} > 0$.
Thus, there are many examples of non-logconcave distributions $\nu$ on $\R^n$ satisfying LSI (with a constant independent of dimension).
There are also Lyapunov function criteria and exponential integrability conditions that can be used to verify when a probability distribution satisfies LSI; see for example~\cite{C04,CM10,MS14,WW16,BGMZ18}.

\subsubsection{Talagrand inequality}

Recall the {\bf Wasserstein distance} between $\rho$ and $\nu$ is
\begin{align}\label{Eq:W2}
W_2(\rho,\nu) = \inf_{\Pi} \E_\Pi[\|X-Y\|^2]^{\frac{1}{2}}
\end{align}
where the infimum is over joint distributions $\Pi$ of $(X,Y)$ with the correct marginals $X \sim \rho, Y \sim \nu$.

Recall we say $\nu$ satisfies {\bf Talagrand inequality} with a constant $\alpha > 0$ if for all $\rho$:
\begin{align}\label{Eq:T}
\frac{\alpha}{2} W_2(\rho,\nu)^2 \le H_\nu(\rho).
\end{align}
Talagrand's inequality implies concentration of measure of Gaussian type.
It was first studied by Talagrand~\cite{T96} for Gaussian $\nu$,
and extended by Otto and Villani~\cite{OV00} to all $\nu$ satisfying LSI; namely, if $\nu$ satisfies LSI with constant $\alpha > 0$, then $\nu$ also satisfies Talagrand's inequality with the same constant~\cite[Theorem~1]{OV00}.
Therefore, under LSI, KL divergence controls the Wasserstein distance.
Moreover, when $\nu$ is log-concave, LSI and Talagrand's inequality are equivalent~\cite[Corollary~3.1]{OV00}.

We recall the geometric interpretation of LSI and Talagrand's inequality from~\cite{OV00}.
In the space of probability distributions with the Riemannian metric defined by the Wasserstein $W_2$ distance, the relative Fisher information~\eqref{Eq:Jnu} is the squared norm of the gradient of KL divergence~\eqref{Eq:Hnu}.
Therefore, LSI~\eqref{Eq:LSI-KL} is the gradient dominated condition (also known as the Polyak-\L{ojaciewicz} (PL) inequality) for KL divergence.
On the other hand, Talagrand's inequality~\eqref{Eq:T} is the quadratic growth condition for KL divergence.
In general, the gradient dominated condition implies the quadratic growth condition~\cite[Proposition~1']{OV00};
therefore, LSI implies Talagrand's inequality.

\subsection{Langevin dynamics}
\label{Sec:Langevin}

The {\bf Langevin dynamics} for target distribution $\nu = e^{-f}$ is a continuous-time stochastic process $(X_t)_{t \ge 0}$ in $\R^n$ that evolves following the stochastic differential equation:
\begin{align}\label{Eq:LD}
dX_t = -\nabla f(X_t) \, dt + \sqrt{2} \, dW_t
\end{align}
where $(W_t)_{t \ge 0}$ is the standard Brownian motion in $\R^n$ with $W_0 = 0$.

If $(X_t)_{t \ge 0}$ evolves following the Langevin dynamics~\eqref{Eq:LD}, then their probability density function $(\rho_t)_{t \ge 0}$ evolves following the {\bf Fokker-Planck equation}:
\begin{align}\label{Eq:FP}
\part{\rho_t}{t} \,=\, \nabla \cdot (\rho_t \nabla f) + \Delta \rho_t \,=\, \nabla \cdot \left(\rho_t \nabla \log \frac{\rho_t}{\nu}\right).
\end{align}
Here $\nabla \cdot$ is the divergence and $\Delta$ is the Laplacian operator.
We provide a derivation in Appendix~\ref{App:FP}.
From~\eqref{Eq:FP}, if $\rho_t = \nu$, then $\part{\rho_t}{t} = 0$, so $\nu$ is the stationary distribution for the Langevin dynamics~\eqref{Eq:LD}.
Moreover, the Langevin dynamics brings any distribution $X_t \sim \rho_t$ closer to the target distribution $\nu$, as the following lemma shows.

\begin{lemma}\label{Lem:Hdot}
Along the Langevin dynamics~\eqref{Eq:LD} (or equivalently, the Fokker-Planck equation~\eqref{Eq:FP}),
\begin{align}\label{Eq:HdotLD}
\frac{d}{dt} H_\nu(\rho_t) = - J_\nu(\rho_t).
\end{align}
\end{lemma}

We provide the proof of Lemma~\ref{Lem:Hdot} in Section~\ref{App:ProofHdot}.
Since $J_\nu(\rho) \ge 0$, the identity~\eqref{Eq:HdotLD} shows that KL divergence with respect to $\nu$ is decreasing along the Langevin dynamics, so indeed the distribution $\rho_t$ converges to $\nu$.

\subsubsection{Exponential convergence of KL divergence along Langevin dynamics under LSI}

When $\nu$ satisfies LSI, KL divergence converges exponentially fast along the Langevin dynamics.

\begin{theorem}\label{Thm:HRate-LSI}
Suppose $\nu$ satisfies LSI with constant $\alpha > 0$.
Along the Langevin dynamics~\eqref{Eq:LD},
\begin{align}\label{Eq:HRateLD}
H_\nu(\rho_t) \le e^{-2\alpha t} H_\nu(\rho_0).
\end{align}
Furthermore, $W_2(\rho_t,\nu) \le \sqrt{\frac{2}{\alpha}H_\nu(\rho_0)}\, e^{-\alpha t}$.
\end{theorem}

We provide the proof of Theorem~\ref{Thm:HRate-LSI} in Section~\ref{App:HRate-LSI}.
We also recall the optimization interpretation of Langevin dynamics as the gradient flow of KL divergence in the space of distributions with the Wasserstein metric~\cite{JKO98,Vil03,OV00}.
Then the exponential convergence rate in Theorem~\ref{Thm:HRate-LSI} is a manifestation of the general fact that gradient flow converges exponentially fast under gradient domination condition.
This provides a justification for using the Langevin dynamics for sampling from $\nu$, as a natural steepest descent flow that minimizes the KL divergence $H_\nu$.

\section{Unadjusted Langevin Algorithm}
\label{Sec:ULA}

In this section we study the behavior of KL divergence along the Unadjusted Langevin Algorithm (ULA) in discrete time under log-Sobolev inequality assumption.

Suppose we wish to sample from a smooth target probability distribution $\nu = e^{-f}$ in $\R^n$.
The {\bf Unadjusted Langevin Algorithm (ULA)} with step size $\step > 0$ is the discrete-time algorithm
\begin{align}\label{Eq:ULA}
x_{k+1} = x_k - \step \nabla f(x_k) + \sqrt{2\step} \, z_k
\end{align}
where $z_k \sim \N(0,I)$ is an independent standard Gaussian random variable in $\R^n$.
Let $\rho_k$ denote the probability distribution of $x_k$ that evolves following ULA.

As $\step \to 0$, ULA recovers the Langevin dynamics~\eqref{Eq:LD} in continuous-time.
However, for fixed $\step > 0$, ULA converges to a biased limiting distribution $\nu_\step \neq \nu$.
Therefore, KL divergence $H_\nu(\rho_k)$ does not tend to $0$ along ULA, as it has an asymptotic bias $H_\nu(\nu_\step) > 0$.

\begin{example}\label{Ex:ULAGaussian}
Let $\nu = \N(0,\frac{1}{\alpha} I)$.
The ULA iteration is $x_{k+1} = (1-\step\alpha)x_k + \sqrt{2\step}z_k$, $z_k \sim \N(0,I)$.
For $0 < \step < \frac{2}{\alpha}$, the limit is $\nu_\step = \N\left(0,\frac{1}{\alpha\left(1-\frac{\step\alpha}{2}\right)}\right)$,
and the bias is
$H_\nu(\nu_\step) = \frac{n}{2}\left(\frac{\step\alpha}{2\left(1-\frac{\step\alpha}{2}\right)} + \log \left(1-\frac{\step\alpha}{2}\right)\right).$
In particular, $H_\nu(\nu_\step) \le \frac{n\step^2\alpha^2}{16\left(1-\frac{\step\alpha}{2}\right)^2}=O(\step^2)$.
\end{example}

\subsection{Convergence of KL divergence along ULA under LSI}
\label{Sec:KL-ULA-LSI}

When the true target distribution $\nu$ satisfies LSI and a smoothness condition, we can prove a convergence guarantee in KL divergence along ULA.
Recall we say $\nu = e^{-f}$ is $L$-smooth, $0 < L < \infty$, if $-LI \preceq \nabla^2 f(x) \preceq LI$ for all $x \in \R^n$.

A key part in our analysis is the following lemma which bounds the decrease in KL divergence along one iteration of ULA.
Here $x_{k+1} \sim \rho_{k+1}$ is the output of one step of ULA~\eqref{Eq:ULA} from $x_k \sim \rho_k$.

\begin{lemma}\label{Lem:OneStep}
Suppose $\nu$ satisfies LSI with constant $\alpha > 0$ and is $L$-smooth.
If $0 < \step \le \frac{\alpha}{4L^2}$, then along each step of ULA~\eqref{Eq:ULA},
\begin{align}\label{Eq:Rec}
H_\nu(\rho_{k+1}) \le e^{-\alpha \step} H_\nu(\rho_k) + 6 \step^2 n L^2.
\end{align}
\end{lemma}

We provide the proof of Lemma~\ref{Lem:OneStep} in Section~\ref{App:ProofULA}.
The proof of Lemma~\ref{Lem:OneStep} compares the evolution of KL divergence along one step of ULA with the evolution along the Langevin dynamics in continuous time (which converges exponentially fast under LSI), and bounds the discretization error; see Figure~\ref{Fig:KLProof} for an illustration.
This high-level comparison technique has been used in many papers. 
Our proof structure is similar to that of Cheng and Bartlett~\cite{CB18}, whose analysis needs $\nu$ to be strongly log-concave.

 \begin{figure}[h!t!]
 \centering
  \includegraphics[width=0.5\textwidth]{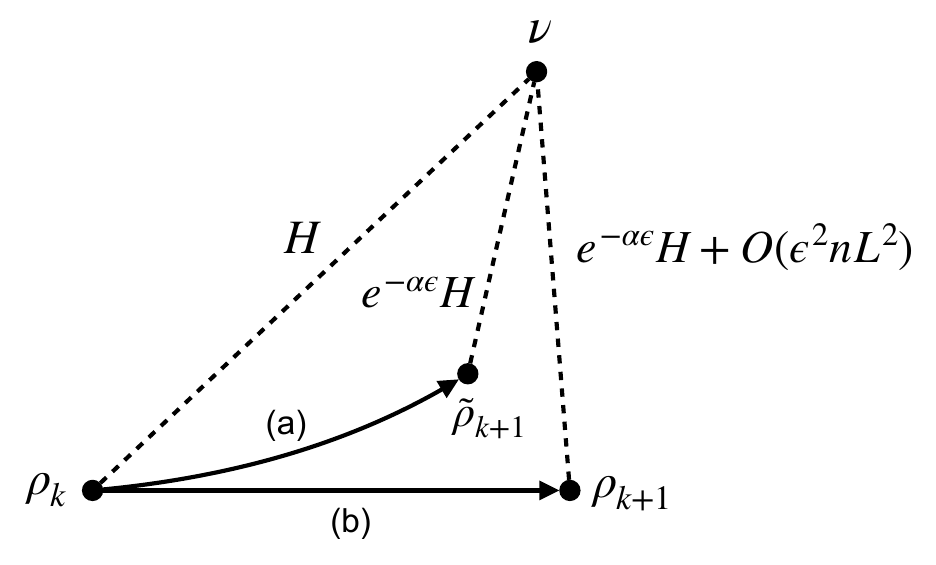}
 \caption{An illustration for the proof of Lemma~\ref{Lem:OneStep}.
 In each iteration, we compare the evolution of (a) the continuous-time Langevin dynamics for time $\step$, and (b) one step of ULA.
 If the current KL divergence is $H \equiv H_\nu(\rho_k)$, then after the Langevin dynamics (a) the KL divergence is $H_\nu(\tilde \rho_{k+1}) \le e^{-\alpha \step} H$, and we show that after ULA (b) the KL divergence is $H_\nu(\rho_{k+1}) \le e^{-\alpha \step} H + O(\step^2 nL^2)$.
 }
 \label{Fig:KLProof}
 \end{figure}

With Lemma~\ref{Lem:OneStep}, we can prove our main result on the convergence rate of ULA under LSI.

\begin{reptheorem}{Thm:Main}
Suppose $\nu$ satisfies LSI with constant $\alpha > 0$ and is $L$-smooth.
For any $x_0 \sim \rho_0$ with $H_\nu(\rho_0) < \infty$, the iterates $x_k \sim \rho_k$ of 
ULA~\eqref{Eq:ULA} with step size $0 < \step \le \frac{\alpha}{4L^2}$ satisfy
\begin{align}\label{Eq:Main1}
H_\nu(\rho_k) \le e^{-\alpha \step k} H_\nu(\rho_0) + \frac{8 \step n L^2}{\alpha}.
\end{align}
Thus, for any $\err > 0$, to achieve $H_\nu(\rho_k) < \err$, it suffices to run ULA with step size $\step \le \frac{\alpha}{4L^2}\min\{1,\frac{\err}{4n}\}$ for $k \ge \frac{1}{\alpha \step} \log \frac{2 H_\nu(\rho_0)}{\err}$ iterations.
\end{reptheorem}

We provide the proof of Theorem~\ref{Thm:Main} in Section~\ref{App:ProofThmMain}.

In particular, suppose $\err < 4n$ and we choose the largest permissible step size $\step = \Theta \left(\frac{\alpha \err}{L^2 n}\right)$.
Suppose we start with a Gaussian $\rho_0 = \N(x^\ast,\frac{1}{L}I)$, where $x^\ast$ is a stationary point of $f$ (which we can find, e.g., via gradient descent), so $H_\nu(\rho_0) \le f(x^\ast) + \frac{n}{2} \log \frac{L}{2\pi} = \tilde O(n)$ by Lemma~\ref{Lem:KLEstimate}.
Therefore, Theorem~\ref{Thm:Main} states that to achieve $H_\nu(\rho_k) \le \err$, ULA has iteration complexity $k = \tilde \Theta\left( \frac{L^2 n}{\alpha^2 \err} \right)$.
Since LSI implies Talagrand's inequality, Theorem~\ref{Thm:Main} also yields a convergence guarantee in Wasserstein distance. 

As $k \to \infty$, Theorem~\ref{Thm:Main} implies the following bound on the bias between $\nu_\step$ and $\nu$ under LSI.
However, note that the bound in Corollary~\ref{Cor:BiasLSI} is $H_\nu(\nu_\step) = O(\step)$, while from Example~\ref{Ex:ULAGaussian} we see that $H_\nu(\nu_\step) = O(\step^2)$ in the Gaussian case.

\begin{corollary}\label{Cor:BiasLSI}
Suppose $\nu$ satisfies LSI with constant $\alpha > 0$ and is $L$-smooth.
For $0 < \step \le \frac{\alpha}{4L^2}$, the biased limit $\nu_\step$ of ULA with step size $\step$ satisfies $H_\nu(\nu_\step) \le \frac{8 nL^2\step}{\alpha}$ and $W_2(\nu,\nu_\step)^2 \le \frac{16nL^2\step}{\alpha^2}$. 
\end{corollary}

\begin{remark}
If $f$ satisfies a third-order smoothness condition (without isoperimetry), then we can show a bound on the bias in relative Fisher information; see Section~\ref{Sec:BiasULA}.
\end{remark}

\section{Review of R\'enyi divergence along Langevin dynamics}
\label{Sec:Renyi}

In this section we review the definition of R\'enyi divergence and the exponential convergence of R\'enyi divergence along the Langevin dynamics under LSI.

\subsection{R\'enyi divergence}

R\'enyi divergence~\cite{R61} is a family of generalizations of KL divergence.
We refer to~\cite{VH14,BCG19} for basic properties of R\'enyi divergence.

For $q > 0$, $q \neq 1$, the {\bf R\'enyi divergence} of order $q$ of a probability distribution $\rho$ with respect to $\nu$ is
\begin{align}\label{Eq:RenyiDef}
R_{q,\nu}(\rho) = \frac{1}{q-1} \log F_{q,\nu}(\rho)
\end{align}
where
\begin{align}
F_{q,\nu}(\rho) = \E_\nu\left[\left(\frac{\rho}{\nu}\right)^q\right] = \int_{\R^n} \nu(x) \frac{\rho(x)^q}{\nu(x)^q} \, dx = \int_{\R^n} \frac{\rho(x)^q}{\nu(x)^{q-1}} dx.
\end{align}
R\'enyi divergence is the relative form of {\em R\'enyi entropy}~\cite{R61}: $H_q(\rho) = \frac{1}{q-1} \log \int \rho(x)^q \, dx$.
The case $q = 1$ is defined via limit, and recovers the KL divergence~\eqref{Eq:Hnu}:
\begin{align}\label{Eq:RenyiDef1}
R_{1,\nu}(\rho) = \lim_{q \to 1} R_{q,\nu}(\rho) = \E_\nu\left[\frac{\rho}{\nu} \log \frac{\rho}{\nu}\right] = \E_\rho\left[\log \frac{\rho}{\nu}\right] = H_\nu(\rho).
\end{align}
R\'enyi divergence has the property that $R_{q,\nu}(\rho) \ge 0$ for all $\rho$, and $R_{q,\nu}(\rho) = 0$ if and only if $\rho = \nu$.
Furthermore, the map $q \mapsto R_{q,\nu}(\rho)$ is increasing (see Section~\ref{App:RenyiProperties}).
Therefore, R\'enyi divergence provides an alternative measure of ``distance'' of $\rho$ from $\nu$, which becomes stronger as $q$ increases.
In particular, $R_{\infty,\nu}(\rho) = \log \left\|\frac{\rho}{\nu}\right\|_\infty = \log \sup_{x} \frac{\rho(x)}{\nu(x)}$ is finite if and only if $\rho$ is {\em warm} relative to $\nu$.
It is possible that $R_{q,\nu}(\rho) = \infty$ for large enough $q$, as the following example shows.

\begin{example}\label{Ex:GaussianRenyi}
Let $\rho = \N(0,\sigma^2 I)$ and $\nu = \N(0, \lambda^2 I)$.
If $\sigma^2 > \lambda^2$ and $q \ge \frac{\sigma^2}{\sigma^2-\lambda^2}$, then $R_{q,\nu}(\rho) = \infty$.
Otherwise, $R_{q,\nu}(\rho) = \frac{n}{2} \log \frac{\lambda^2}{\sigma^2} - \frac{n}{2(q-1)} \log\left(q - (q-1) \frac{\sigma^2}{\lambda^2}\right)$.
\end{example}

Analogous to Lemma~\ref{Lem:KLEstimate}, we have the following estimate of the R\'enyi divergence of a Gaussian.

\begin{lemma}\label{Lem:GaussianRenyi}
Suppose $\nu = e^{-f}$ is $L$-smooth.
Let $\rho = \N(x^\ast, \frac{1}{L} I)$ where $x^\ast$ is a stationary point of $f$.
Then for all $q \ge 1$, $R_{q,\nu}(\rho) \le f(x^\ast) + \frac{n}{2} \log \frac{L}{2\pi}$.
\end{lemma}

We provide the proof of Lemma~\ref{Lem:GaussianRenyi} in Section~\ref{App:GaussianRenyi}.

\subsubsection{Log-Sobolev inequality}

For $q > 0$, we define the {\bf R\'enyi information} of order $q$ of $\rho$ with respect to $\nu$ as
\begin{align}\label{Eq:InfoDef}
G_{q,\nu}(\rho) 
= \E_\nu\left[\left(\frac{\rho}{\nu}\right)^q \left\|\nabla \log \frac{\rho}{\nu} \right\|^2\right] 
= \E_\nu\left[\left(\frac{\rho}{\nu}\right)^{q-2} \left\|\nabla \frac{\rho}{\nu} \right\|^2\right] 
= \frac{4}{q^2} \E_\nu\left[\left\|\nabla \left(\frac{\rho}{\nu}\right)^{\frac{q}{2}}\right\|^2\right].
\end{align}
The case $q=1$ recovers relative Fisher information~\eqref{Eq:Jnu}:
$G_{1,\nu}(\rho) = \E_\nu\left[\frac{\rho}{\nu} \left\|\nabla \log \frac{\rho}{\nu} \right\|^2\right] 
= J_\nu(\rho).$
We have the following relation under log-Sobolev inequality.
Note that the case $q=1$ recovers LSI in the form~\eqref{Eq:LSI-KL} involving KL divergence and relative Fisher information. 

\begin{lemma}\label{Lem:RenyiLSI}
Suppose $\nu$ satisfies LSI with constant $\alpha > 0$.
Let $q \ge 1$.
For all $\rho$,
\begin{align}\label{Eq:RenyiLSI}
\frac{G_{q,\nu}(\rho)}{F_{q,\nu}(\rho)} \ge \frac{2\alpha}{q^2} R_{q,\nu}(\rho).
\end{align}
\end{lemma}

We provide the proof of Lemma~\ref{Lem:RenyiLSI} in Section~\ref{App:RenyiLSI}.

\subsection{Langevin dynamics}
\label{Sec:RenyiLDRate}

Along the Langevin dynamics~\eqref{Eq:LD} for $\nu$, we can compute the rate of change of the R\'enyi divergence.

\begin{lemma}\label{Lem:RenyiLD}
For all $q > 0$, along the Langevin dynamics~\eqref{Eq:LD},
\begin{align}\label{Eq:RenyiLD}
\frac{d}{dt} R_{q,\nu}(\rho_t) = -q \frac{G_{q,\nu}(\rho_t)}{F_{q,\nu}(\rho_t)}.
\end{align}
\end{lemma}

We provide the proof of Lemma~\ref{Lem:RenyiLD} in Section~\ref{App:RenyiLD}.
In particular, $\frac{d}{dt} R_{q,\nu}(\rho_t) \le 0$, so R\'enyi divergence is always decreasing along the Langevin dynamics.
Furthermore, analogous to how the Langevin dynamics is the gradient flow of KL divergence under the Wasserstein metric, one can also show that the Langevin dynamics is the  the gradient flow of R\'enyi divergence with respect to a suitably defined metric (which depends on the target distribution $\nu$) on the space of distributions; see~\cite{CLL18}.

\subsubsection{Convergence of R\'enyi divergence along Langevin dynamics under LSI}

When $\nu$ satisfies LSI, R\'enyi divergence converges exponentially fast along the Langevin dynamics.
Note the case $q=1$ recovers the exponential convergence rate of KL divergence from Theorem~\ref{Thm:HRate-LSI}.

\begin{reptheorem}{Thm:RenyiLD}
Suppose $\nu$ satisfies LSI with constant $\alpha > 0$.
Let $q \ge 1$.
Along the Langevin dynamics~\eqref{Eq:LD},
\begin{align}\label{Eq:RenyiLDRate}
R_{q,\nu}(\rho_t) \le e^{-\frac{2\alpha t}{q}} R_{q,\nu}(\rho_0).
\end{align}
\end{reptheorem}

We provide the proof of Theorem~\ref{Thm:RenyiLD} in Section~\ref{App:ProofThmRenyiLD}.
Theorem~\ref{Thm:RenyiLD} shows that if the initial R\'enyi divergence is finite, then it converges exponentially fast.
However, even if initially the R\'enyi divergence of some order is infinite, it will be eventually finite along the Langevin dynamics, after which time Theorem~\ref{Thm:RenyiLD} applies.
This is because when $\nu$ satisfies LSI, the Langevin dynamics satisfies a {\em hypercontractivity} property~\cite{G75,BGL01,Vil03}; see Section~\ref{App:Hypercontractivity}.
Furthermore, as shown in~\cite{CLL18}, we can combine the exponential convergence rate above with the hypercontractivity property to improve the exponential rate to be $2\alpha$, independent of $q$, at the cost of some initial waiting time; here we leave the rate as above for simplicity.

\begin{remark}\label{Rem:PoincareLangevin}
When $\nu$ satisfies Poincar\'e inequality, we can still prove the convergence of R\'enyi divergence along the Langevin dynamics.
However, in this case, R\'enyi divergence initially decreases linearly, then exponentially once it is less than $1$.
See Section~\ref{Sec:PoincareLangevin}.
\end{remark}

\section{R\'enyi divergence along ULA}
\label{Sec:RenyiULA}

In this section we prove a convergence guarantee for R\'enyi divergence along ULA under the assumption that the biased limit satisfies LSI.

As before, let $\nu = e^{-f}$, and let $\nu_\step$ denote the biased limit of ULA~\eqref{Eq:ULA} with step size $\step > 0$.
We first note that the asymptotic bias $R_{q,\nu}(\nu_\step)$ may be infinite for large enough $q$. 

\begin{example}\label{Ex:ULAGaussian2}
As in Examples~\ref{Ex:ULAGaussian} and~\ref{Ex:GaussianRenyi}, let
$\nu = \N(0,\frac{1}{\alpha}I)$, so $\nu_\step = \N\left(0,\frac{1}{\alpha\left(1-\frac{\step\alpha}{2}\right)}\right)$.
The bias is
\begin{align*}
R_{q,\nu}(\nu_\step) =
\begin{cases}
  \frac{n}{2(q-1)} \left(q\log\left(1-\frac{\step\alpha}{2}\right)-\log\left(1-\frac{q\step\alpha}{2}\right)\right) ~ & \text{ if } 1 < q < \frac{2}{\step\alpha}, \\
  \infty & \text{ if } q \ge \frac{2}{\step\alpha}.
\end{cases}
\end{align*}
For $1 < q < \frac{2}{\step \alpha}$, we can bound 
$R_{q,\nu}(\nu_\step) \le \frac{n\alpha^2q^2\step^2}{8(q-1)\left(1-\frac{q\step\alpha}{2}\right)}.$
\end{example}

Thus, for each fixed $q > 1$, there is an asymptotic bias $R_{q,\nu}(\nu_\step)$ which is finite for small $\step > 0$.
In Example~\ref{Ex:ULAGaussian2}, we have $R_{q,\nu}(\nu_\step) = O(\step^2)$.

\subsection{Decomposition of R\'enyi divergence}

For order $q > 1$, we have the following decomposition of R\'enyi divergence.

\begin{lemma}\label{Lem:RenyiDecomp}
Let $q > 1$.
For all probability distribution $\rho$,
\begin{align}\label{Eq:RenyiDecomp}
R_{q,\nu}(\rho) \le \left(\frac{q-\frac{1}{2}}{q-1}\right) R_{2q,\nu_\step}(\rho) + R_{2q-1,\nu}(\nu_\step).
\end{align}
\end{lemma}

We provide the proof of Lemma~\ref{Lem:RenyiDecomp} in Section~\ref{App:RenyiDecomp}.
The first term in the bound above is the R\'enyi divergence with respect to the biased limit, which converges exponentially fast under LSI assumption (see Lemma~\ref{Lem:RenyiRateLSI}).
The second term in~\eqref{Eq:RenyiDecomp} is the asymptotic bias in R\'enyi divergence.

\subsection{Rapid convergence of R\'enyi divergence to biased limit under LSI}
\label{Sec:RenyiBiased}

We show that R\'enyi divergence with respect to the biased limit $\nu_\step$ converges exponentially fast along ULA, assuming $\nu_\step$ itself satisfies LSI.

\begin{assumption}\label{As:RenyiLSI}
The probability distribution $\nu_\step$ satisfies LSI with a constant $\beta \equiv \beta_\step > 0$.
\end{assumption}

We can verify Assumption~\ref{As:RenyiLSI} in the Gaussian case.
We can also verify Assumption~\ref{As:RenyiLSI} when $\nu$ is smooth and strongly log-concave; see Section~\ref{Sec:BiasLSI}.
However, it is unclear how to verify Assumption~\ref{As:RenyiLSI} in general.
One might hope to prove that if $\nu$ satisfies LSI, then Assumption~\ref{As:RenyiLSI} holds.

\begin{example}
Let $\nu = \N(0,\frac{1}{\alpha}I)$,
so $\nu_\step = \N\Big(0,\frac{1}{\alpha\left(1-\frac{\step\alpha}{2}\right)} I \Big)$,
which is strongly log-concave (and hence satisfies LSI) with parameter $\beta = \alpha\left(1-\frac{\step\alpha}{2}\right)$.
In particular, $\beta \ge \frac{\alpha}{2}$ for $\step \le \frac{1}{\alpha}$.
\end{example}

Under Assumption~\ref{As:RenyiLSI}, we can prove an exponential convergence rate to the biased limit $\nu_\step$.

\begin{lemma}\label{Lem:RenyiRateLSI}
Assume Assumption~\ref{As:RenyiLSI}.
Suppose $\nu = e^{-f}$ is $L$-smooth, and let $0 < \step \le \min\left\{\frac{1}{3L}, \frac{1}{9\beta}\right\}$.
For $q \ge 1$,
along ULA~\eqref{Eq:ULA},
\begin{align}\label{Eq:RenyiRateLSI}
R_{q,\nu_\step}(\rho_k) \le e^{-\frac{\beta \step k}{q}} R_{q,\nu_\step}(\rho_0).
\end{align}
\end{lemma}

We provide the proof of Lemma~\ref{Lem:RenyiRateLSI} in Section~\ref{App:RenyiRateLSI}.
In the proof of Lemma~\ref{Lem:RenyiRateLSI}, we decompose each step of ULA as a sequence of two operations; see Figure~\ref{Fig:RenyiProof} for an illustration.
In the first part, we take a gradient step; this is a deterministic bijective map, so it preserves R\'enyi divergence.
In the second part, we add an independent Gaussian; this is the result of evolution along the heat flow, and we can derive a formula on the decrease in R\'enyi divergence (which is similar to the formula~\eqref{Eq:RenyiLD} along the Langevin dynamics;
see Section~\ref{App:RenyiRateLSI} for detail).

 \begin{figure}[h!t!]
 \centering
  \includegraphics[width=0.4\textwidth]{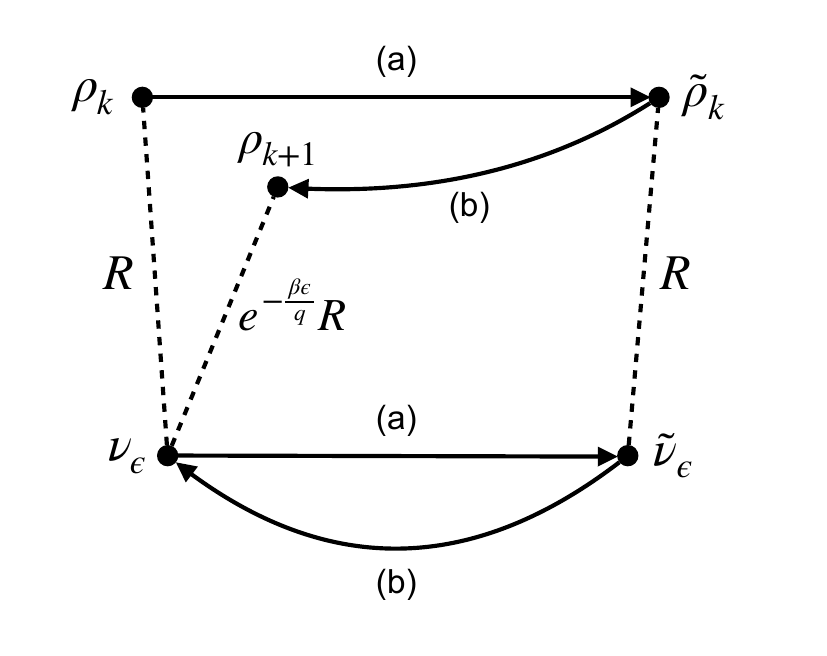}
 \caption{An illustration for the proof of Lemma~\ref{Lem:RenyiRateLSI}.
 We decompose each step of ULA into two operations: (a) a deterministic gradient step, and (b) an evolution along the heat flow. 
 If the current R\'enyi divergence is $R \equiv R_{q,\nu_\step}(\rho_k)$, then the gradient step (a) does not change the R\'enyi divergence: $R_{q,\tilde \nu_\step}(\tilde \rho_k) = R$, while the heat flow (b) decreases the R\'enyi divergence: $R_{q,\nu_\step}(\rho_{k+1}) \le e^{-\alpha \step} R$.
 }
 \label{Fig:RenyiProof}
 \end{figure}

\subsection{Convergence of R\'enyi divergence along ULA under LSI}
\label{Sec:RenyiULALSI}

We combine Lemma~\ref{Lem:RenyiDecomp} and Lemma~\ref{Lem:RenyiRateLSI} to obtain the following characterization of the convergence of R\'enyi divergence along ULA under LSI.

\begin{theorem}\label{Thm:RenyiRate}
Assume Assumption~\ref{As:RenyiLSI}.
Suppose $\nu = e^{-f}$ is $L$-smooth, and let $0 < \step \le \min\left\{\frac{1}{3L}, \frac{1}{9\beta}\right\}$.
Let $q > 1$, and suppose $R_{2q,\nu_\step}(\rho_0) < \infty$.
Then along ULA~\eqref{Eq:ULA},
\begin{align}
R_{q,\nu}(\rho_k) \le \left(\frac{q-\frac{1}{2}}{q-1}\right) R_{2q,\nu_\step}(\rho_0) e^{-\frac{\beta \step k}{2q}} + R_{2q-1,\nu}(\nu_\step).
\end{align}
\end{theorem}

We provide the proof of Theorem~\ref{Thm:RenyiRate} in Section~\ref{App:ProofThmRenyiRate}.
For $\err > 0$, let $\gamma_q(\err)= \sup \{ \step > 0 \colon R_{q,\nu}(\nu_\step) \le \err \}$.
Theorem~\ref{Thm:RenyiRate} states that to achieve $R_{q,\nu}(\rho_k) \le \err$, it suffices to run ULA with step size $\step = \Theta\left(\min\left\{\frac{1}{L},\gamma_{2q-1}\left(\frac{\err}{2}\right)\right\}\right)$ for $k = \Theta\left(\frac{1}{\beta \step} \log \frac{R_{2q,\nu_\step}(\rho_0)}{\err}\right)$ iterations.
Suppose $\err$ is small enough that $\gamma_{2q-1}\left(\frac{\err}{2}\right) < \frac{1}{L}$.
Note that $\nu_\step$ is $\frac{1}{2\step}$-smooth, so by choosing $\rho_0$ to be a Gaussian with covariance $2\step I$, we have $R_{2q,\nu_\step}(\rho_0) = \tilde O(n)$ by Lemma~\ref{Lem:GaussianRenyi}.
Therefore, Theorem~\ref{Thm:RenyiRate} yields an iteration complexity of $k = \tilde \Theta\left(\frac{1}{\beta \gamma_{2q-1}\left(\frac{\err}{2}\right)}\right)$.

For example, if $R_{q,\nu}(\nu_\step) = O(\step)$, then $\gamma_q(\err) = \Omega(\err)$, so the iteration complexity is $k = \tilde \Theta\left(\frac{1}{\beta \err}\right)$ with step size $\step = \Theta(\err)$.
On the other hand, if $R_{q,\nu}(\step) = O(\step^2)$, as in Example~\ref{Ex:ULAGaussian2}, then $\gamma_q(\err) = \Omega(\sqrt{\err})$, so the iteration complexity is $k = \tilde \Theta\left(\frac{1}{\beta \sqrt{\err}}\right)$ with step size $\step = \Theta(\sqrt{\err})$.

\begin{remark}
Our result for R\'enyi divergence above involves the asymptotic bias, which we do not bound.
Another approach to analyze ULA in R\'enyi divergence was proposed in~\cite{GT20} (and improved in~\cite{EHZ20}), albeit with a bound that does not provide an estimate of the R\'enyi bias.
The work of~\cite{CELSZ21} extended our one-step interpolation technique to show the convergence of ULA in R\'enyi divergence under LSI and smoothness, and provides an estimate on the R\'enyi bias.
\end{remark}

\section{Poincar\'e inequality}
\label{Sec:Poincare}

In this section we review the definition of Poincar\'e inequality and prove convergence guarantees for the R\'enyi divergence along the Langevin dynamics and ULA.
As before, let $\rho, \nu$ be smooth probability distributions on $\R^n$.

Recall we say $\nu$ satisfies {\bf Poincar\'e inequality (PI)} with a constant $\alpha > 0$ if for all smooth function $g \colon \R^n \to \R$,
\begin{align}\label{Eq:PI}
\Var_\nu(g) \le \frac{1}{\alpha} \E_\nu[\|\nabla g\|^2]
\end{align}
where $\Var_\nu(g) = \E_\nu[g^2] - \E_\nu[g]^2$ is the variance of $g$ under $\nu$.
Poincar\'e inequality is an isoperimetric-type statement, but it is weaker than LSI.
It is known that LSI implies PI with the same constant; in fact, PI is a linearization of LSI~\eqref{Eq:LSI-KL}, i.e., when $\rho = (1+\eta g) \nu$ as $\eta \to 0$~\cite{R81,Vil03}.
Furthermore, it is also known that Talagrand's inequality implies PI with the same constant, and in fact PI is also a linearization of Talagrand's inequality~\cite{OV00}.
Poincar\'e inequality is better behaved than LSI~\cite{CM10}, and there are various Lyapunov function criteria and integrability conditions that can be used to verify when a probability distribution satisfies Poincar\'e inequality; see for example~\cite{BBCG08,MS14,C18}.

Under Poincar\'e inequality, we can prove the following bound on R\'enyi divergence, which is analogous to Lemma~\ref{Lem:RenyiLSI} under LSI. When $R_{q,\nu}(\rho)$ is small, the two bounds are approximately equivalent.

\begin{lemma}\label{Lem:RenyiPI}
Suppose $\nu$ satisfies Poincar\'e inequality with constant $\alpha > 0$.
Let $q \ge 2$.
For all $\rho$,
\begin{align*}
\frac{G_{q,\nu}(\rho)}{F_{q,\nu}(\rho)} \ge \frac{4\alpha}{q^2} \left(1-e^{-R_{q,\nu}(\rho)}\right). 
\end{align*}
\end{lemma}

We provide the proof of Lemma~\ref{Lem:RenyiPI} in Section~\ref{App:RenyiPI}.

\subsection{Convergence of R\'enyi divergence along Langevin dynamics under Poincar\'e}
\label{Sec:PoincareLangevin}

When $\nu$ satisfies Poincar\'e inequality, R\'enyi divergence converges along the Langevin dynamics.
The convergence is initially linear, then becomes exponential once the R\'enyi divergence falls below a constant.

\begin{reptheorem}{Thm:RenyiLP}
Suppose $\nu$ satisfies Poincar\'e inequality with constant $\alpha > 0$.
Let $q \ge 2$.
Along the Langevin dynamics~\eqref{Eq:LD},
\begin{align*}
R_{q,\nu}(\rho_t) \le
\begin{cases}
R_{q,\nu}(\rho_0) -\frac{2\alpha t}{q} ~~ & \text{ if } R_{q,\nu}(\rho_0) \ge 1 \text{ and as long as } R_{q,\nu}(\rho_t) \ge 1, \\
e^{-\frac{2\alpha t}{q}} R_{q,\nu}(\rho_0) ~~ & \text{ if } R_{q,\nu}(\rho_0) \le 1.
\end{cases}
\end{align*}
\end{reptheorem}

We provide the proof of Theorem~\ref{Thm:RenyiLP} in Section~\ref{App:RenyiLP}.
Theorem~\ref{Thm:RenyiLP} states that starting from $R_{q,\nu}(\rho_0) \ge 1$, the Langevin dynamics reaches $R_{q,\nu}(\rho_t) \le \err$ in $t \le O\left(\frac{q}{\alpha} \left(R_{q,\nu}(\rho_0) + \log \frac{1}{\err}\right)\right)$ time.

\subsection{Convergence of R\'enyi divergence to biased limit under Poincar\'e}
\label{Sec:PoincareULA}

We show that R\'enyi divergence with respect to the biased limit $\nu_\step$ converges exponentially fast along ULA, assuming $\nu_\step$ satisfies Poincar\'e inequality.

\begin{assumption}\label{As:RenyiP}
The distribution $\nu_\step$ satisfies Poincar\'e inequality with a constant $\beta \equiv \beta_\step > 0$.
\end{assumption}

We can verify Assumption~\ref{As:RenyiP} in the Gaussian case, and when $\nu$ is smooth and strongly log-concave; see Section~\ref{Sec:BiasLSI}.
However, it is unclear how to verify Assumption~\ref{As:RenyiP} in general.
One might hope to prove that if $\nu$ satisfies Poincar\'e, then Assumption~\ref{As:RenyiP} holds.

Analogous to Lemma~\ref{Lem:RenyiRateLSI}, we have the following convergence to the biased limit in discrete time, at a rate which matches the continuous-time convergence in Theorem~\ref{Thm:RenyiRatePoincare}.

\begin{lemma}\label{Lem:RenyiRateP}
Assume Assumption~\ref{As:RenyiP}.
Suppose $\nu = e^{-f}$ is $L$-smooth, and let $0 < \step \le \min\left\{\frac{1}{3L}, \frac{1}{9\beta}\right\}$.
For $q \ge 2$,
along ULA~\eqref{Eq:ULA},
\begin{align}\label{Eq:RenyiRateP}
R_{q,\nu_\step}(\rho_k) \le
\begin{cases}
R_{q,\nu_\step}(\rho_0) -\frac{\beta \step k}{q} ~~ & \text{ if } R_{q,\nu_\step}(\rho_0) \ge 1 \text{ and as long as } R_{q,\nu_\step}(\rho_k) \ge 1, \\
e^{-\frac{\beta \step k}{q}} R_{q,\nu_\step}(\rho_0) ~~ & \text{ if } R_{q,\nu_\step}(\rho_0) \le 1.
\end{cases}
\end{align}
\end{lemma}

We provide the proof of Lemma~\ref{Lem:RenyiRateP} in Section~\ref{App:RenyiRateP}.
Lemma~\ref{Lem:RenyiRateP} states that starting from $R_{q,\nu_\step}(\rho_0) \ge 1$, ULA reaches $R_{q,\nu_\step}(\rho_k) \le \err$ in $k \le O\left(\frac{q}{\step \beta} \left(R_{q,\nu_\step}(\rho_0) + \log \frac{1}{\err}\right)\right)$ iterations.

\subsection{Convergence of R\'enyi divergence along ULA under Poincar\'e}
\label{Sec:RenyiRatePoincare}

We combine Lemma~\ref{Lem:RenyiDecomp} and Lemma~\ref{Lem:RenyiRateP} to obtain the following characterization of the convergence of R\'enyi divergence along ULA to the true target distribution under Poincar\'e inequality.

\begin{theorem}\label{Thm:RenyiRatePoincare}
Assume Assumption~\ref{As:RenyiP}.
Suppose $\nu = e^{-f}$ is $L$-smooth, and let $0 < \step \le \min\left\{\frac{1}{3L}, \frac{1}{9\beta}\right\}$.
Let $q > 1$, and suppose $1 \le R_{2q,\nu_\step}(\rho_0) < \infty$.
Then along ULA~\eqref{Eq:ULA}, for $k \ge k_0 := \frac{2q}{\beta\step}(R_{2q,\nu_\step}(\rho_0)-1)$,
\begin{align}
R_{q,\nu}(\rho_k) \le \left(\frac{q-\frac{1}{2}}{q-1}\right) e^{-\frac{\beta \step (k-k_0)}{2q}} + R_{2q-1,\nu}(\nu_\step).
\end{align}
\end{theorem}

We provide the proof of Theorem~\ref{Thm:RenyiRatePoincare} in Section~\ref{App:RenyiRatePoincare}.

For $\err > 0$, recall $\gamma_q(\err) = \sup \{ \step > 0 \colon R_{q,\nu}(\nu_\step) \le \err \}$.
Theorem~\ref{Thm:RenyiRatePoincare} states that to achieve $R_{q,\nu}(\rho_k) \le \err$, it suffices to run ULA with step size $\step = \Theta\left(\min\left\{\frac{1}{L},\gamma_{2q-1}\left(\frac{\err}{2}\right)\right\}\right)$ for $k = \Theta\left(\frac{1}{\beta \step} \left( R_{2q,\nu_\step}(\rho_0) + \log \frac{1}{\err}\right)\right)$ iterations.
Suppose $\err$ is small enough that $\gamma_{2q-1}\left(\frac{\err}{2}\right) < \frac{1}{L}$.
Note that $\nu_\step$ is $\frac{1}{2\step}$-smooth, so by choosing $\rho_0$ to be a Gaussian with covariance $2\step I$, we have $R_{2q,\nu_\step}(\rho_0) = \tilde O(n)$ by Lemma~\ref{Lem:GaussianRenyi}.
Therefore, Theorem~\ref{Thm:RenyiRatePoincare} yields an iteration complexity of $k = \tilde \Theta\left(\frac{n}{\beta \gamma_{2q-1}\left(\frac{\err}{2}\right)}\right)$.
Note the additional dependence on dimension, compared to the LSI case in Section~\ref{Sec:RenyiULALSI}.

For example, if $R_{q,\nu}(\nu_\step) = O(\step)$, then $\gamma_q(\err) = \Omega(\err)$, so the iteration complexity is $k = \tilde \Theta\left(\frac{n}{\beta \err}\right)$ with step size $\step = \Theta(\err)$.
On the other hand, if $R_{q,\nu}(\nu_\step) = O(\step^2)$, as in Example~\ref{Ex:ULAGaussian2}, then $\gamma_q(\err) = \Omega(\sqrt{\err})$, so the iteration complexity is $k = \tilde \Theta\left(\frac{n}{\beta \sqrt{\err}}\right)$ with step size $\step = \Theta(\sqrt{\err})$.

\section{Properties of Biased Limit}
\label{Sec:Bias}

\subsection{Bound on bias under third-order smoothness}
\label{Sec:BiasULA}

Let $\nu_\step$ be the biased limit of ULA with step size $\step > 0$.
Let 
$\mu_{\step} = (I - \step \nabla f)_\# \nu_\step$,
so $\nu_\step$ satisfies
\begin{align*}
    \nu_\step = \mu_{\step} \ast \N(0,2\step I).
\end{align*}

We will bound bound the relative Fisher information $J_\nu(\nu_\step)$ under third-order smoothness.
We say $f$ is {\bf $(L,M)$-smooth} if $f$ is $L$-smooth ($\nabla f$ is $L$-Lipschitz), and $\nabla^2 f$ is $M$-Lipschitz, or $\|\nabla^3 f\|_\op \le M$.
We provide the proof of Theorem~\ref{Thm:Bias} in Sections~\ref{Sec:ProofBiasUpper} and~\ref{Sec:ProofBiasLower}.

\begin{theorem}\label{Thm:Bias}
\begin{enumerate}
    \item If $f$ is $(L,M)$-smooth and $\step \le \frac{1}{2L}$, then:
    \begin{align*}
        J_\nu(\nu_\step) \le  2 \step n \left(L^2 + 2 \sqrt{nL}M + 3nM^2\right).
    \end{align*}
        
    \item For any $f$ and $\step > 0$ (such that $\nu_\step$ exists and the quantities below are defined):
    \begin{align*}
        J_\nu(\nu_\step) \ge \frac{\step^2}{4} \frac{\big( \E_{\nu_\step}[\|\nabla f\|^2] \big)^2}{\Var_{\nu_\step}(X)}.
    \end{align*}
\end{enumerate}
\end{theorem}

Note the dependence on $\step$ in the upper bound above is $O(\eta^2)$, while the lower bound is $\Omega(\eta^2)$.

\begin{example}
Recall that if $\nu = \N(0,\frac{1}{\alpha} I)$, then $\nu_\step = \N(0, \frac{1}{\alpha (1-\frac{\step \alpha}{2})} I)$ for $\step < \frac{2}{\alpha}$. 
Then\footnote{
Recall for $\nu = \N(0,\frac{1}{\alpha} I)$ and 
$\rho = \N(0,\frac{1}{\beta} I)$,
the relative Fisher information is $J_\nu(\rho) = \frac{n}{\beta} (\beta-\alpha)^2$.
}
\begin{align}
J_\nu(\nu_\step) = 
\frac{\step^2}{4} \frac{n\alpha^2}{(1-\frac{\step \alpha}{2})^2} = \Theta(\step^2 n \alpha^2).
\end{align}
So the lower bound in Theorem~\ref{Thm:Bias} has the right order of $\step$, but not the upper bound.
\end{example}

\begin{remark}
Recall from Theorem~\ref{Thm:Main} that under LSI and $L$-smoothness we have $H_\nu(\nu_\step) \le O(\step n L^2)$.
Under LSI, the upper bound in Theorem~\ref{Thm:Bias} 
implies $H_\nu(\nu_\step) \le O(\step n (L^2 + \sqrt{nL} M + n M^2))$, which has an additional dependence on third-order smoothness, but without requiring isoperimetry.
\end{remark}

However, we also note that in general, convergence in relative Fisher information does not necessarily imply convergence of the underlying distributions; see for example~\cite{BCESZ22}.

\begin{remark}
By examining the proof of the upper bound in Section~\ref{Sec:ProofBiasUpper}, we can also conclude that $J_\nu(\nu_\step) \le \step n L^2$ assuming $f$ is $L$-smooth and $\Delta \Delta f \ge 0$. \end{remark}

\subsection{Isoperimetry of biased limit under strong log-concavity and smoothness}
\label{Sec:BiasLSI}

If $\nu$ is smooth and strongly log-concave, then the biased limit $\nu_\step$ satisfies LSI (hence also Poincar\'e), so Assumptions~\ref{As:RenyiLSI} and~\ref{As:RenyiP} are satisfied. 
The authors thank Sinho Chewi for communicating the following result to us.
We provide the proof of Theorem~\ref{Thm:BiasLSI} in Section~\ref{Sec:ProofBiasLSI}.

\begin{theorem}\label{Thm:BiasLSI}
If $\nu$ is $\alpha$-strongly log-concave and $L$-smooth, and $\step \le \frac{1}{L}$, then $\nu_\step$ is $\beta$-LSI with $\beta \ge \frac{\alpha}{2}$.
\end{theorem}

With Theorem~\ref{Thm:BiasLSI}, we know that for target distributions which are smooth and strongly log-concave, we have convergence of ULA in R\'enyi divergence to the biased limit, as in Theorem~\ref{Thm:RenyiRate}.
However, the final bound is in terms of the bias in R\'enyi divergence, which we do not bound.
(Under third-order smoothness, we can bound it in relative Fisher information as in Theorem~\ref{Thm:Bias}, but it does not bound the R\'enyi divergence.)
The work of~\cite{BCESZ22} extends our interpolation technique to show the convergence in R\'enyi divergence under LSI as well as a general family of isoperimetric inequalities, and proves a bound on the R\'enyi bias under LSI and smoothness.

\section{Proofs and details}
\label{Sec:Proofs}

\subsection{Proofs for \S\ref{Sec:Review}: KL divergence along Langevin dynamics}
\label{App:Review}

\subsubsection{Proof of Lemma~\ref{Lem:KLEstimate}}
\label{App:ProofLemKLEstimate}

\begin{proof}[Proof of Lemma~\ref{Lem:KLEstimate}]
Since $f$ is $L$-smooth and $\nabla f(x^\ast) = 0$, we have the bound
\begin{align*}
f(x) \le f(x^\ast) + \langle \nabla f(x^\ast),x-x^\ast \rangle + \frac{L}{2} \|x-x^\ast\|^2 = f(x^\ast) + \frac{L}{2}\|x-x^\ast\|^2.
\end{align*}
Let $X \sim \rho = \N(x^\ast,\frac{1}{L}I)$. 
Then
\begin{align*}
\E_\rho[f(X)] \le f(x^\ast) + \frac{L}{2} \Var_\rho(X) = f(x^\ast) + \frac{n}{2}.
\end{align*}
Recall the entropy of $\rho$ is $H(\rho) = -\E_\rho[\log \rho(X)] = \frac{n}{2} \log \frac{2\pi e}{L}$.
Therefore, the KL divergence is
\begin{align*}
H_\nu(\rho) = \int \rho \left(\log \rho + f\right) dx = -H(\rho) + \E_\rho[f] \le f(x^\ast) + \frac{n}{2} \log \frac{L}{4\pi e}.
\end{align*}
\end{proof}

\subsubsection{Proof of Lemma~\ref{Lem:Hdot}}
\label{App:ProofHdot}

\begin{proof}[Proof of Lemma~\ref{Lem:Hdot}]
Recall the time derivative of KL divergence along any flow is given by
\begin{align*}
\frac{d}{dt} H_\nu(\rho_t) &= \frac{d}{dt} \int_{\R^n} \rho_t \log \frac{\rho_t}{\nu} \, dx
= \int_{\R^n} \part{\rho_t}{t} \log \frac{\rho_t}{\nu} \, dx
\end{align*}
since the second part of the chain rule is zero:
$\int \rho_t \frac{\partial}{\partial t} \log \frac{\rho_t}{\nu} \, dx 
= \int \part{\rho_t}{t} \, dx
= \frac{d}{dt} \int \rho_t \, dx
= 0.$
Therefore, along the Fokker-Planck equation~\eqref{Eq:FP} for the Langevin dynamics~\eqref{Eq:LD}, 
\begin{align*}
\frac{d}{dt} H_\nu(\rho_t) &= \int \nabla \cdot \left( \rho_t \nabla \log \frac{\rho_t}{\nu} \right) \, \log \frac{\rho_t}{\nu} \, dx \\
&= -\int \rho_t \left\| \nabla \log \frac{\rho_t}{\nu} \right\|^2 dx \\
&= -J_\nu(\rho_t)
\end{align*}
where in the second equality we have applied integration by parts.
\end{proof}

\subsubsection{Proof of Theorem~\ref{Thm:HRate-LSI}}
\label{App:HRate-LSI}

\begin{proof}[Proof of Theorem~\ref{Thm:HRate-LSI}]
From Lemma~\ref{Lem:Hdot} and the LSI assumption~\eqref{Eq:LSI-KL},
\begin{align*}
\frac{d}{dt} H_\nu(\rho_t) = -J_\nu(\rho_t) \le -2\alpha H_\nu(\rho_t).
\end{align*}
Integrating implies the desired bound $H_\nu(\rho_t) \le e^{-2\alpha t} H_\nu(\rho_0)$.

Furthermore, since $\nu$ satisfies LSI with constant $\alpha$, it also satisfies Talagrand's inequality~\eqref{Eq:T} with constant $\alpha$~\cite[Theorem~1]{OV00}.
Therefore, $W_2(\rho_t,\nu)^2 \le \frac{2}{\alpha} H_\nu(\rho_t) \le \frac{2}{\alpha} e^{-2\alpha t} H_\nu(\rho_0)$, as desired.
\end{proof}

\subsection{Proofs for \S\ref{Sec:ULA}: Unadjusted Langevin Algorithm}
\label{App:ULA}

\subsubsection{Proof of Lemma~\ref{Lem:OneStep}}
\label{App:ProofULA}

We will use the following auxiliary results.

\begin{lemma}\label{Lem:GradStat}
Assume $\nu = e^{-f}$ is $L$-smooth.
Then
\begin{align*}
\E_\nu[\|\nabla f\|^2] \le nL.
\end{align*}
\end{lemma}
\begin{proof}
Since $\nu = e^{-f}$, by integration by parts we can write
\begin{align*}
\E_\nu[\|\nabla f\|^2] = \E_\nu[\Delta f].
\end{align*}
Since $\nu$ is $L$-smooth, $\nabla^2 f(x) \preceq L \, I$, so $\Delta f(x) \le nL$ for all $x \in \R^n$.
Therefore,
$\E_\nu[\|\nabla f\|^2] = \E_\nu[\Delta f] \le nL$,
as desired.
\end{proof}

\begin{lemma}\label{Lem:Grad}
Suppose $\nu$ satisfies Talagrand's inequality with constant $\alpha > 0$ and is $L$-smooth.
For any $\rho$,
\begin{align*}
\E_\rho[\|\nabla f\|^2] \le \frac{4L^2}{\alpha} H_\nu(\rho) + 2nL.
\end{align*}
\end{lemma}
\begin{proof}
Let $x \sim \rho$ and $x^\ast \sim \nu$ with an optimal coupling $(x,x^\ast)$ so that $\E[\|x-x^\ast\|^2] = W_2(\rho,\nu)^2$.
Since $\nu = e^{-f}$ is $L$-smooth, $\nabla f$ is $L$-Lipschitz.
By triangle inequality,
\begin{align*}
\|\nabla f(x)\| &\le \|\nabla f(x) - \nabla f(x^\ast)\| + \|\nabla f(x^\ast)\| \\
&\le L\|x-x^\ast\| + \|\nabla f(x^\ast)\|.
\end{align*}
Squaring, using $(a+b)^2 \le 2a^2 + 2b^2$, and taking expectation, we get
\begin{align*}
\E_\rho[\|\nabla f(x)\|^2] &\le 2L^2 \, \E[\|x-x^\ast\|^2] + 2\E_\nu[\|\nabla f(x^\ast)\|^2] \\
&= 2L^2 \, W_2(\rho,\nu)^2 + 2\E_\nu[\|\nabla f(x^\ast)\|^2].
\end{align*}
By Talagrand's inequality~\eqref{Eq:T}, $W_2(\rho,\nu)^2 \le \frac{2}{\alpha} H_\nu(\rho)$.
By Lemma~\ref{Lem:GradStat} we have $\E_\nu[\|\nabla f(x^\ast)\|^2] \le nL$.
Plugging these to the bound above gives the desired result.
\end{proof}

We are now ready to prove Lemma~\ref{Lem:OneStep}.

\begin{proof}[Proof of Lemma~\ref{Lem:OneStep}]
For simplicity suppose $k=0$, so we start at $x_0 \sim \rho_0$.
We write one step of ULA 
\begin{align*}
x_0 \mapsto x_0 - \step \nabla f(x_0) + \sqrt{2\step} z_0
\end{align*}
as the output at time $\step$ of the stochastic differential equation
\begin{align}\label{Eq:ULAStoch}
dx_t = -\nabla f(x_0) \, dt + \sqrt{2} \, dW_t
\end{align}
where $W_t$ is the standard Brownian motion in $\R^n$ starting at $W_0 = 0$.
Indeed, the solution to~\eqref{Eq:ULAStoch} at time $t = \step$ is
\begin{align}
x_\step &= x_0 - \step \nabla f(x_0) + \sqrt{2} \, W_\step \notag \\ 
&\stackrel{d}{=} \, x_0 - \step \nabla f(x_0) + \sqrt{2\step} \, z_0. \label{Eq:ULASoln}
\end{align}
where $z_0 \sim \N(0,I)$, which is identical to the ULA update.

We derive the continuity equation corresponding to~\eqref{Eq:ULAStoch} as follows.
For each $t > 0$, let $\rho_{0t}(x_0,x_t)$ denote the joint distribution of $(x_0,x_t)$, which we write in terms of the conditionals and marginals as
\begin{align*}
\rho_{0t}(x_0,x_t) \,=\, \rho_0(x_0) \rho_{t|0}(x_t\,|\,x_0) \,=\, \rho_t(x_t) \rho_{0|t}(x_0\,|\,x_t).
\end{align*}
Conditioning on $x_0$, the drift vector field $-\nabla f(x_0)$ is a constant, so the Fokker-Planck formula for the conditional density $\rho_{t|0}(x_t\,|\,x_0)$ is
\begin{align}\label{Eq:FPCond}
\part{\rho_{t|0}(x_t\,|\,x_0)}{t} = \nabla \cdot \left(\rho_{t|0}(x_t\,|\,x_0) \nabla f(x_0)\right) + \Delta \rho_{t|0}(x_t\,|\,x_0).
\end{align}
To derive the evolution of $\rho_t$, we take expectation over $x_0 \sim \rho_0$.
Multiplying both sides of~\eqref{Eq:FPCond} by $\rho_0(x_0)$ and integrating over $x_0$, we obtain
\begin{align}
\part{\rho_{t}(x)}{t} 
&= \int_{\R^n} \part{\rho_{t|0}(x\,|\,x_0)}{t} \, \rho_0(x_0) \, dx_0 \notag \\
&= \int_{\R^n} \left( \nabla \cdot \left(\rho_{t|0}(x\,|\,x_0) \nabla f(x_0)\right) + \Delta \rho_{t|0}(x\,|\,x_0) \right) \rho_0(x_0) \, dx_0 \notag \\
&= \int_{\R^n} \left( \nabla \cdot \left(\rho_{t,0}(x, x_0) \nabla f(x_0)\right) + \Delta \rho_{t,0}(x, x_0) \right) dx_0 \notag \\
&= \nabla \cdot \left( \rho_t(x) \int_{\R^n} \rho_{0|t}(x_0\,|\,x) \nabla f(x_0) \, dx_0 \right) + \Delta \rho_{t}(x) \notag \\
&= \nabla \cdot \left( \rho_t(x) \E_{\rho_{0|t}}[\nabla f(x_0) \,|\, x_t = x]\right) + \Delta \rho_t(x).  \label{Eq:ULAPDE}
\end{align}
Observe that the difference between the Fokker-Planck equations~\eqref{Eq:ULAPDE} for ULA and~\eqref{Eq:FP} for Langevin dynamics is in the first term, that the drift is now the conditional expectation $\E_{\rho_{0|t}}[\nabla f(x_0) \,|\, x_t = x]$, rather than the true gradient $\nabla f(x)$.

Recall the time derivative of relative entropy along any flow is given by
\begin{align*}
\frac{d}{dt} H_\nu(\rho_t) &= \frac{d}{dt} \int_{\R^n} \rho_t \log \frac{\rho_t}{\nu} \, dx
= \int_{\R^n} \part{\rho_t}{t} \log \frac{\rho_t}{\nu} \, dx
\end{align*}
since the second part of the chain rule is zero:
$\int \rho_t \frac{\partial}{\partial t} \log \frac{\rho_t}{\nu} \, dx 
= \int \part{\rho_t}{t} \, dx
= \frac{d}{dt} \int \rho_t \, dx
= 0.$

Therefore, the time derivative of relative entropy for ULA, using the Fokker-Planck equation~\eqref{Eq:ULAPDE} and integrating by parts, is given by: 
\begin{align}
\frac{d}{dt} H_\nu(\rho_t) &= \int_{\R^n} \left(\nabla \cdot \left( \rho_t(x) \E_{\rho_{0|t}}[\nabla f(x_0) \,|\, x_t = x]\right) + \Delta \rho_t(x) \right) \, \log \frac{\rho_t(x)}{\nu(x)} \, dx \notag \\
&= \int_{\R^n} \left(\nabla \cdot \left( \rho_t(x) \left( \nabla \log \frac{\rho_t(x)}{\nu(x)} + \E_{\rho_{0|t}}[\nabla f(x_0) \,|\, x_t = x] - \nabla f(x) \right) \right) \right) \, \log \frac{\rho_t(x)}{\nu(x)} \, dx \notag \\
&= -\int_{\R^n} \rho_t(x) \left\langle \nabla \log \frac{\rho_t(x)}{\nu(x)} + \E_{\rho_{0|t}}[\nabla f(x_0) \,|\, x_t = x] - \nabla f(x), \, \nabla \log \frac{\rho_t(x)}{\nu(x)} \right\rangle \, dx \notag \\
&= -\int_{\R^n} \rho_t(x) \left\|\nabla \log \frac{\rho_t}{\nu} \right\|^2 dx + \int_{\R^n} \rho_t(x) \left\langle \nabla f(x) - \E_{\rho_{0|t}}[\nabla f(x_0) \,|\, x_t = x], \, \nabla \log \frac{\rho_t(x)}{\nu(x)} \right\rangle dx \notag  \\
&= -J_\nu(\rho_t) + \int_{\R^n \times \R^n} \rho_{0t}(x_0,x) \left\langle \nabla f(x) - \nabla f(x_0), \nabla \log \frac{\rho_t(x)}{\nu(x)} \right\rangle dx_0 \, dx \notag \\
&= -J_\nu(\rho_t) + \E_{\rho_{0t}}\left[ \left\langle \nabla f(x_t) - \nabla f(x_0), \, \nabla \log \frac{\rho_t(x_t)}{\nu(x_t)} \right\rangle \right] \label{Eq:ULAHdot}
\end{align}
where in the last step we have renamed $x$ as $x_t$. 
The first term in~\eqref{Eq:ULAHdot} is the same as in the Langevin dynamics.
The second term in~\eqref{Eq:ULAHdot} is the discretization error, which we can bound as follows.
Using $\langle a,b \rangle \le \|a\|^2 + \frac{1}{4} \|b\|^2$ and since $\nabla f$ is $L$-Lipschitz,
\begin{align}
\E_{\rho_{0t}}\left[\left\langle \nabla f(x_t) - \nabla f(x_0), \nabla \log \frac{\rho_t(x_t)}{\nu(x_t)} \right\rangle \right] 
&\le \E_{\rho_{0t}}[\|\nabla f(x_t)-\nabla f(x_0)\|^2] + \frac{1}{4} \E_{\rho_{0t}}\left[\left\|\nabla \log \frac{\rho_t(x_t)}{\nu(x_t)}\right\|^2\right] \notag \\
&= \E_{\rho_{0t}}[\|\nabla f(x_t)-\nabla f(x_0)\|^2] + \frac{1}{4} J_\nu(\rho_t)  \notag \\
&\le L^2 \E_{\rho_{0t}}[\|x_t-x_0\|^2] + \frac{1}{4} J_\nu(\rho_t) \label{Eq:Comp1}
\end{align}
Recall from~\eqref{Eq:ULASoln} the solution of ULA is $x_t \stackrel{d}{=} x_0 - t\nabla f(x_0) + \sqrt{2t} \, z_0$, 
where $z_0 \sim \N(0,I)$ is independent of $x_0$.
Then
\begin{align*}
\E_{\rho_{0t}}[\|x_t-x_0\|^2] &= \E_{\rho_{0t}}[\|-t \nabla f(x_0) + \sqrt{2t} z_0\|^2] \\
&= t^2 \E_{\rho_0}[\|\nabla f(x_0)\|^2] + 2tn \\
&\le \frac{4 t^2 L^2}{\alpha} H_\nu(\rho_0) + 2t^2 nL + 2tn
\end{align*}
where in the last inequality we have used Lemma~\ref{Lem:Grad}.
This bounds the discretization error by
\begin{align*}
\E_{\rho_{0t}}\left[\left\langle \nabla f(x_t) - \nabla f(x_0), \nabla \log \frac{\rho_t(x_t)}{\nu(x_t)} \right\rangle \right]  \le 
\frac{4 t^2 L^4}{\alpha} H_\nu(\rho_0) + 2t^2 nL^3 + 2tn L^2 + \frac{1}{4} J_\nu(\rho_t).
\end{align*}

Therefore, from~\eqref{Eq:ULAHdot}, the time derivative of KL divergence along ULA is bounded by
\begin{align*}
\frac{d}{dt} H_\nu(\rho_t) \le -\frac{3}{4} J_\nu(\rho_t) + \frac{4 t^2 L^4}{\alpha} H_\nu(\rho_0) + 2t^2 nL^3 + 2tn L^2.
\end{align*}
Then by the LSI~\eqref{Eq:LSI-KL} assumption,
\begin{align*}
\frac{d}{dt} H_\nu(\rho_t) 
\le -\frac{3\alpha}{2} H_\nu(\rho_t) + \frac{4 t^2 L^4}{\alpha} H_\nu(\rho_0) + 2t^2 nL^3 + 2tn L^2.
\end{align*}
We wish to integrate the inequality above for $0 \le t \le \step$.
Using $t \le \step$ and since $\step \le \frac{1}{2L}$, we simplify the above to
\begin{align*}
\frac{d}{dt} H_\nu(\rho_t) &\le -\frac{3\alpha}{2} H_\nu(\rho_t) + \frac{4 \step^2 L^4}{\alpha} H_\nu(\rho_0) + 2\step^2 nL^3 + 2\step n L^2 \\
&\le -\frac{3\alpha}{2} H_\nu(\rho_t) + \frac{4 \step^2 L^4}{\alpha} H_\nu(\rho_0) + 3\step n L^2.
\end{align*}
Multiplying both sides by $e^{\frac{3\alpha}{2} t}$, we can write the above as
\begin{align*}
\frac{d}{dt} \left(e^{\frac{3\alpha}{2} t} H_\nu(\rho_t) \right) &\le e^{\frac{3\alpha}{2} t} \left(\frac{4 \step^2 L^4}{\alpha} H_\nu(\rho_0) + 3\step n L^2 \right).
\end{align*}
Integrating from $t=0$ to $t=\step$ gives
\begin{align*}
e^{\frac{3}{2} \alpha\step} H_\nu(\rho_\step) - H_\nu(\rho_0) &\le \frac{2(e^{\frac{3}{2} \alpha\step}-1)}{3\alpha} \left(\frac{4 \step^2 L^4}{\alpha} H_\nu(\rho_0) + 3\step n L^2 \right) \\
&\le 2\step \left(\frac{4 \step^2 L^4}{\alpha} H_\nu(\rho_0) + 3\step n L^2 \right)
\end{align*}
where in the last step we have used the inequality $e^c \le 1+2c$ for $0 < c = \frac{3}{2} \alpha \step \le 1$, which holds because $0 < \step \le \frac{2}{3\alpha}$.
Rearranging, the inequality above gives
\begin{align*}
H_\nu(\rho_\step) \le e^{-\frac{3}{2} \alpha \step} \left(1+ \frac{8\step^3 L^4}{\alpha}\right) H_\nu(\rho_0) + e^{-\frac{3}{2} \alpha \step} 6\step^2 nL^2.
\end{align*}
Since $1+ \frac{8\step^3 L^4}{\alpha} \le 1+\frac{\alpha \step}{2} \le e^{\frac{1}{2} \alpha \step}$ for $\step \le \frac{\alpha}{4L^2}$, and using $e^{-\frac{3}{2} \alpha \step} \le 1$, we conclude that
\begin{align*}
H_\nu(\rho_\step) \le e^{-\alpha \step} H_\nu(\rho_0) + 6\step^2 nL^2.
\end{align*}
This is the desired inequality, after renaming $\rho_0 \equiv \rho_k$ and $\rho_\step \equiv \rho_{k+1}$.
Note that the conditions $\step \le \frac{1}{2L}$ and $\step \le \frac{2}{3\alpha}$ above are also implied by the assumption $\step \le \frac{\alpha}{4L^2}$ since $\alpha \le L$.
\end{proof}

\subsubsection{Proof of Theorem~\ref{Thm:Main}}
\label{App:ProofThmMain}

\begin{proof}[Proof of Theorem~\ref{Thm:Main}]
Applying the recursion~\eqref{Eq:Rec} from Lemma~\ref{Lem:OneStep}, we obtain
\begin{align*}
H_\nu(\rho_k) 
\,\le\, e^{-\alpha \step k} H_\nu(\rho_0) + \frac{6 \step^2 n L^2}{1-e^{-\alpha \step}}
\,\le\, e^{-\alpha \step k} H_\nu(\rho_0) + \frac{8 \step n L^2}{\alpha}
\end{align*}
where in the last step we have used the inequality $1-e^{-c} \ge \frac{3}{4} c$ for $0 < c = \alpha \step \le \frac{1}{4}$, which holds since $\step \le \frac{\alpha}{4 L^2} \le \frac{1}{4\alpha}$.

Given $\err > 0$, if we further assume $\step \le \frac{\err \alpha}{16 n L^2}$, then the above implies
$H_\nu(\rho_k) \le e^{-\alpha \step k} H_\nu(\rho_0) + \frac{\err}{2}.$
This means for $k \ge \frac{1}{\alpha \step} \log \frac{2 H_\nu(\rho_0)}{\err}$, we have $H_\nu(\rho_k) \le \frac{\err}{2} + \frac{\err}{2} = \err$, as desired.
\end{proof}

\subsection{Details for \S\ref{Sec:Renyi}: R\'enyi divergence along Langevin dynamics}
\label{App:Renyi}

\subsubsection{Properties of R\'enyi divergence}
\label{App:RenyiProperties}

We recall that R\'enyi divergence is increasing in the order.

\begin{lemma}\label{Lem:RenyiIncreasing}
For any probability distributions $\rho,\nu$,
the map $q \mapsto R_{q,\nu}(\rho)$ is increasing for $q > 0$.
\end{lemma}
\begin{proof}
Let $0 < q \le r$. 
We will show that $R_{q,\nu}(\rho) \le R_{r,\nu}(\rho)$.

First suppose $q > 1$.
We write $F_{q,\nu}(\rho)$ as an expectation over $\rho$ and use power mean inequality:
\begin{align*}
F_{q,\nu}(\rho) = \E_{\nu}\left[\left(\frac{\rho}{\nu}\right)^{q}\right] 
= \E_{\rho}\left[\left(\frac{\rho}{\nu}\right)^{q-1}\right]
\le \E_{\rho}\left[\left(\frac{\rho}{\nu}\right)^{r-1}\right]^{\frac{q-1}{r-1}}
= \E_{\nu}\left[\left(\frac{\rho}{\nu}\right)^{r}\right]^{\frac{q-1}{r-1}}
= F_{r,\nu}(\rho)^{\frac{q-1}{r-1}}.
\end{align*}
Taking logarithm and dividing by $q-1 > 0$ gives
\begin{align*}
R_{q,\nu}(\rho) = \frac{1}{q-1} \log F_{q,\nu}(\rho) \le \frac{1}{r-1} \log F_{r,\nu}(\rho) = R_{r,\nu}(\rho).
\end{align*}
The case $q=1$ follows by taking limit $q \to 1$.

Now suppose $q \le r < 1$, so $1-q \ge 1-r > 0$.
We again write $F_{q,\nu}(\rho)$ as an expectation over $\rho$ and use power mean inequality:
\begin{align*}
F_{q,\nu}(\rho) = \E_{\nu}\left[\left(\frac{\rho}{\nu}\right)^{q}\right] 
= \E_{\rho}\left[\left(\frac{\nu}{\rho}\right)^{1-q}\right]
\ge \E_{\rho}\left[\left(\frac{\nu}{\rho}\right)^{1-r}\right]^{\frac{1-q}{1-r}}
= \E_{\nu}\left[\left(\frac{\rho}{\nu}\right)^{r}\right]^{\frac{1-q}{1-r}}
= F_{r,\nu}(\rho)^{\frac{1-q}{1-r}}.
\end{align*}
Taking logarithm and dividing by $q-1 < 0$ (which flips the inequality) gives
\begin{align*}
R_{q,\nu}(\rho) = \frac{1}{q-1} \log F_{q,\nu}(\rho) \le \frac{1}{r-1} \log F_{r,\nu}(\rho) = R_{r,\nu}(\rho).
\end{align*}
The case $q < 1 \le r$ follows since $R_{q,\nu}(\rho) \le R_{1,\nu}(\rho) \le R_{r,\nu}(\rho)$.
\end{proof}

\subsubsection{Proof of Lemma~\ref{Lem:GaussianRenyi}}
\label{App:GaussianRenyi}

\begin{proof}[Proof of Lemma~\ref{Lem:GaussianRenyi}]
Since $f$ is $L$-smooth and $x^\ast$ is a stationary point of $f$,
for all $x \in \R^n$ we have
\begin{align*}
f(x) &\le f(x^\ast) + \langle \nabla f(x^\ast), x-x^\ast \rangle + \frac{L}{2}\|x-x^\ast\|^2
= f(x^\ast) + \frac{L}{2}\|x-x^\ast\|^2.
\end{align*}
Let $q > 1$.
Then for $\rho = \N(x^\ast,\sigma^2 I)$ with $\frac{q}{\sigma^2} > (q-1)L$,
\begin{align*}
F_{q,\nu}(\rho) &= \int_{\R^n} \frac{\rho(x)^q}{\nu(x)^{q-1}} dx \\
&= \frac{1}{(2\pi \sigma^2)^{\frac{nq}{2}}} \int_{\R^n} e^{-\frac{q}{2\sigma^2}\|x-x^\ast\|^2 + (q-1)f(x)} dx \\
&\le \frac{1}{(2\pi \sigma^2)^{\frac{nq}{2}}} \int_{\R^n} e^{-\frac{q}{2\sigma^2}\|x-x^\ast\|^2 + (q-1)f(x^\ast) + \frac{(q-1)L}{2}\|x-x^\ast\|^2} dx \\
&= \frac{e^{(q-1)f(x^\ast)}}{(2\pi \sigma^2)^{\frac{nq}{2}}} \int_{\R^n} e^{-\frac{1}{2} \left(\frac{q}{\sigma^2} - (q-1)L\right)\|x-x^\ast\|^2} dx \\
&= \frac{e^{(q-1)f(x^\ast)}}{(2\pi \sigma^2)^{\frac{nq}{2}}} \left(\frac{2\pi}{\frac{q}{\sigma^2} - (q-1)L}\right)^{\frac{n}{2}} \\
&= \frac{e^{(q-1)f(x^\ast)}}{(2\pi)^{\frac{n}{2}(q-1)} (\sigma^2)^{\frac{nq}{2}}} \frac{1}{\left(\frac{q}{\sigma^2} - (q-1)L\right)^{\frac{n}{2}}}.
\end{align*}
Therefore,
\begin{align*}
R_{q,\nu}(\rho) = \frac{1}{q-1} \log F_{q,\nu}(\rho)
\le f(x^\ast) - \frac{n}{2} \log 2\pi - \frac{n}{2(q-1)} \log \sigma^{2q}\left(\frac{q}{\sigma^2} - (q-1)L\right).
\end{align*}
In particular, if $\sigma^2 = \frac{1}{L}$, then $\frac{q}{\sigma^2} - (q-1)L = L > 0$, and the bound above becomes
\begin{align*}
R_{q,\nu}(\rho) \le f(x^\ast) + \frac{n}{2} \log \frac{L}{2\pi}.
\end{align*}
The case $q=1$ follows from Lemma~\ref{Lem:KLEstimate}, since $\frac{1}{4\pi e} < \frac{1}{2\pi}$.
\end{proof}

\subsubsection{Proof of Lemma~\ref{Lem:RenyiLSI}}
\label{App:RenyiLSI}

\begin{proof}[Proof of Lemma~\ref{Lem:RenyiLSI}]
We plug in $h^2 = \left(\frac{\rho}{\nu}\right)^q$ to the LSI definition~\eqref{Eq:LSI} to obtain
\begin{align}
\frac{q^2}{2\alpha} G_{q,\nu}(\rho) 
&\ge q \E_\nu\left[\left(\frac{\rho}{\nu}\right)^q \log \frac{\rho}{\nu} \right] - F_{q,\nu}(\rho) \log F_{q,\nu}(\rho) \label{Eq:RenyiLSICalc} \\
&= q \part{}{q} F_{q,\nu}(\rho) - F_{q,\nu}(\rho) \log F_{q,\nu}(\rho).  \notag
\end{align}
Therefore, 
\begin{align*}
\frac{q^2}{2\alpha} \frac{G_{q,\nu}(\rho)}{F_{q,\nu}(\rho)}
 &\ge q \part{}{q} \log F_{q,\nu}(\rho) - \log F_{q,\nu}(\rho) \\
 &= q \part{}{q} \left((q-1) R_{q,\nu}(\rho) \right) - (q-1) R_{q,\nu}(\rho) \\
 &= q R_{q,\nu}(\rho) + q(q-1) \part{}{q} R_{q,\nu}(\rho)  - (q-1) R_{q,\nu}(\rho) \\
 &= R_{q,\nu}(\rho) + q(q-1) \part{}{q} R_{q,\nu}(\rho) \\
 &\ge R_{q,\nu}(\rho)
\end{align*}
where in the last inequality we have used $q \ge 1$ and $\part{}{q} R_{q,\nu}(\rho) \ge 0$ since $q \mapsto R_{q,\nu}(\rho)$ is increasing by Lemma~\ref{Lem:RenyiIncreasing}.
\end{proof}

\subsubsection{Proof of Lemma~\ref{Lem:RenyiLD}}
\label{App:RenyiLD}

\begin{proof}[Proof of Lemma~\ref{Lem:RenyiLD}]
Let $q > 0$, $q \neq 1$.
By the Fokker-Planck formula~\eqref{Eq:FP} and integration by parts,
\begin{align}
\frac{d}{dt} F_{q,\nu}(\rho_t) 
&= \int_{\R^n} \nu \frac{\part{}{t} (\rho_t^q)}{\nu^q} dx \notag \\
&= q \int_{\R^n} \frac{\rho_t^{q-1}}{\nu^{q-1}} \part{\rho_t}{t} dx \notag \\
&= q \int_{\R^n} \left(\frac{\rho_t}{\nu}\right)^{q-1} \nabla \cdot \left(\rho_t \nabla \log \frac{\rho_t}{\nu} \right) dx \notag \\
&= -q \int_{\R^n} \rho_t \left \langle \nabla \left(\frac{\rho_t}{\nu}\right)^{q-1}, \nabla \log \frac{\rho_t}{\nu} \right\rangle dx \notag \\
&= -q(q-1) \int_{\R^n} \rho_t \left \langle \left(\frac{\rho_t}{\nu}\right)^{q-2} \nabla \frac{\rho_t}{\nu}, \left(\frac{\rho_t}{\nu}\right)^{-1} \nabla \frac{\rho_t}{\nu} \right\rangle dx \notag \\
&= -q(q-1) \E_\nu\left[ \left(\frac{\rho_t}{\nu}\right)^{q-2} \left\|\nabla \frac{\rho_t}{\nu}\right\|^2\right] \notag \\
&= -q(q-1) G_{q,\nu}(\rho_t).  \label{Eq:Fdot}
\end{align}
Therefore,
\begin{align*}
\frac{d}{dt} R_{q,\nu}(\rho_t) = \frac{1}{q-1} \frac{\frac{d}{dt} F_{q,\nu}(\rho_t)}{F_{q,\nu}(\rho_t)}
= -q \frac{G_{q,\nu}(\rho_t)}{F_{q,\nu}(\rho_t)}.
\end{align*}
For $q = 1$, we have $R_{1,\nu}(\rho_t) = H_\nu(\rho_t)$, $G_{1,\nu}(\rho_t) = J_\nu(\rho_t)$, and $F_{1,\nu}(\rho_t) = 1$, and the claim~\eqref{Eq:RenyiLD} follows from Lemma~\ref{Lem:Hdot}.
\end{proof}

\subsubsection{Proof of Theorem~\ref{Thm:RenyiLD}}
\label{App:ProofThmRenyiLD}

\begin{proof}[Proof of Theorem~\ref{Thm:RenyiLD}]
By Lemma~\ref{Lem:RenyiLSI} and Lemma~\ref{Lem:RenyiLD},
\begin{align*}
\frac{d}{dt} R_{q,\nu}(\rho_t)
= -q \frac{G_{q,\nu}(\rho_t)}{F_{q,\nu}(\rho_t)} \le -\frac{2\alpha}{q} R_{q,\nu}(\rho_t).
\end{align*}
Integrating gives
\begin{align*}
R_{q,\nu}(\rho_t) \le e^{-\frac{2\alpha}{q}t} R_{q,\nu}(\rho_0)
\end{align*}
as desired.
\end{proof}

\subsubsection{Hypercontractivity}
\label{App:Hypercontractivity}

\begin{lemma}\label{Lem:Hypercontractivity}
Suppose $\nu$ satisfies LSI with constant $\alpha > 0$.
Let $q_0 > 1$, and suppose $R_{q_0,\nu}(\rho_0) < \infty$.
Define $q_t = 1 + e^{2\alpha t}(q_0-1)$.
Along the Langevin dynamics~\eqref{Eq:LD}, for all $t \ge 0$,
\begin{align}\label{Eq:Hypercontractivity}
\left(1-\frac{1}{q_t}\right) R_{q_t,\nu}(\rho_t) \le \left(1-\frac{1}{q_0}\right)  R_{q_0,\nu}(\rho_0).
\end{align}
In particular, for any $q \ge q_0$, we have $R_{q,\nu}(\rho_t) \le R_{q_0,\nu}(\rho_0) < \infty$ for all $t \ge \frac{1}{2\alpha} \log \frac{q-1}{q_0-1}$.
\end{lemma}
\begin{proof}
We will show $\frac{d}{dt} \left\{ \left(1-\frac{1}{q_t}\right) R_{q_t,\nu}(\rho_t)\right\} \le 0$, which implies the desired relation~\eqref{Eq:Hypercontractivity}.
Since $q_t = 1 + e^{2\alpha t}(q_0-1)$, we have $\dot q_t = \frac{d}{dt} q_t = 2\alpha (q_t-1)$.
Note that
\begin{align*}
\frac{d}{dt} R_{q_t,\nu}(\rho_t)
  &= \frac{d}{dt} \left(\frac{\log F_{q_t,\nu}(\rho_t)}{q_t-1}\right) \\
  &\stackrel{\eqref{Eq:Fdot}}{=} -\frac{\dot q_t \log F_{q_t,\nu}(\rho_t)}{(q_t-1)^2} + \frac{\dot q_t \E_\nu\left[\left(\frac{\rho_t}{\nu}\right)^{q_t} \log \frac{\rho_t}{\nu}\right] - q_t (q_t-1) G_{q_t,\nu}(\rho_t)}{(q_t-1) F_{q_t,\nu}(\rho_t)} \\
  &= -2\alpha R_{q_t,\nu}(\rho_t) + 2\alpha \frac{\E_\nu\left[\left(\frac{\rho_t}{\nu}\right)^{q_t} \log \frac{\rho_t}{\nu}\right]}{F_{q_t,\nu}(\rho_t)} - q_t \frac{G_{q_t,\nu}(\rho_t)}{F_{q_t,\nu}(\rho_t)}.
\end{align*}
In the second equality above we have used our earlier calculation~\eqref{Eq:Fdot} which holds for fixed $q$.
Then by LSI in the form~\eqref{Eq:RenyiLSICalc}, we have
\begin{align*}
\frac{d}{dt} R_{q_t,\nu}(\rho_t) 
  &\le -2\alpha R_{q_t,\nu}(\rho_t) + 2\alpha \left(\frac{q_t}{2\alpha} \frac{G_{q_t,\nu}(\rho_t)}{F_{q_t,\nu}(\rho_t)} + \frac{1}{q_t} \log F_{q_t,\nu}(\rho_t)\right) - q_t \frac{G_{q_t,\nu}(\rho_t)}{F_{q_t,\nu}(\rho_t)} \\
  &= -2\alpha R_{q_t,\nu}(\rho_t) + 2\alpha\left(1-\frac{1}{q_t}\right) R_{q_t,\nu}(\rho_t) \\
  &= -\frac{2\alpha}{q_t} R_{q_t,\nu}(\rho_t).
\end{align*}
Therefore,
\begin{align*}
\frac{d}{dt} \left\{ \left(1-\frac{1}{q_t}\right) R_{q_t,\nu}(\rho_t)\right\}
  &= \frac{\dot q_t}{q_t^2} R_{q_t,\nu}(\rho_t) + \left(1-\frac{1}{q_t}\right) \frac{d}{dt} R_{q_t,\nu}(\rho_t) \\
  &\le \frac{2\alpha(q_t-1)}{q_t^2} R_{q_t,\nu}(\rho_t) - \left(1-\frac{1}{q_t}\right) \frac{2\alpha}{q_t} R_{q_t,\nu}(\rho_t) \\
  &= 0,
\end{align*}
as desired.

Now given $q \ge q_0$, let $t_0 = \frac{1}{2\alpha} \log \frac{q-1}{q_0-1}$ so $q_{t_0} = q$.
Then $R_{q,\nu}(\rho_{t_0}) \le \frac{q}{(q-1)} \frac{(q_0-1)}{q_0} R_{q_0,\nu}(\rho_0) \le R_{q_0,\nu}(\rho_0) < \infty$.
For $t > t_0$, by applying Theorem~\ref{Thm:RenyiLD} starting from $\rho_{t_0}$, we obtain
$R_{q,\nu}(\rho_t) \le e^{-\frac{2\alpha}{q}(t-t_0)} R_{q,\nu}(\rho_{t_0}) \le R_{q,\nu}(\rho_{t_0}) \le R_{q_0,\nu}(\rho_0) < \infty$.
\end{proof}

By combining Theorem~\ref{Thm:RenyiLD} and Lemma~\ref{Lem:Hypercontractivity}, we obtain the following characterization of the behavior of Renyi divergence along the Langevin dynamics under LSI.

\begin{corollary}\label{Cor:RenyiLD}
Suppose $\nu$ satisfies LSI with constant $\alpha > 0$.
Suppose $\rho_0$ satisfies $R_{q_0,\nu}(\rho_0) < \infty$ for some $q_0 > 1$.
Along the Langevin dynamics~\eqref{Eq:LD}, for all $q \ge q_0$ and $t \ge t_0 := \frac{1}{2\alpha} \log \frac{q-1}{q_0-1}$,
\begin{align}\label{Eq:RenyiLDRateFinal}
R_{q,\nu}(\rho_t) \le e^{-\frac{2\alpha}{q}(t-t_0)} R_{q_0,\nu}(\rho_0).
\end{align}
\end{corollary}
\begin{proof}
By Lemma~\ref{Lem:Hypercontractivity}, at $t = t_0$ we have $R_{q,\nu}(\rho_{t_0}) \le R_{q_0,\nu}(\rho_0)$.
For $t > t_0$, by applying Theorem~\ref{Thm:RenyiLD} starting from $\rho_{t_0}$, we have
$R_{q,\nu}(\rho_t) \le e^{-\frac{2\alpha}{q}(t-t_0)} R_{q,\nu}(\rho_{t_0}) \le e^{-\frac{2\alpha}{q}(t-t_0)} R_{q_0,\nu}(\rho_0)$.
\end{proof}

\subsection{Proofs for \S\ref{Sec:RenyiULA}: R\'enyi divergence along ULA}
\label{App:RenyiULA}

\subsubsection{Proof of Lemma~\ref{Lem:RenyiDecomp}}
\label{App:RenyiDecomp}

\begin{proof}[Proof of Lemma~\ref{Lem:RenyiDecomp}]
By Cauchy-Schwarz inequality,
\begin{align*}
F_{q,\nu}(\rho)
  &= \int \frac{\rho^q}{\nu^{q-1}} \, dx \\
  &= \int \nu_\step \left(\frac{\rho}{\nu_\step} \right)^q \left(\frac{\nu_\step}{\nu}\right)^{q-1} dx \\
  &\le \left(\int \nu_\step \left( \frac{\rho}{\nu_\step} \right)^{2q} dx \right)^{\frac{1}{2}} \left(\int \nu_\step \left( \frac{\nu_\step}{\nu} \right)^{2(q-1)} dx \right)^{\frac{1}{2}} \\
  &= F_{2q,\nu_\step}(\rho)^{\frac{1}{2}} F_{2q-1,\nu}(\nu_\step)^{\frac{1}{2}}.
\end{align*}
Taking logarithm gives
\begin{align*}
(q-1) R_{q,\nu}(\rho) \le \frac{(2q-1)}{2} R_{2q,\nu_\step}(\rho) + \frac{(2q-2)}{2} R_{2q-1,\nu}(\nu_\step).
\end{align*}
Dividing both sides by $q-1>0$ gives the desired inequality~\eqref{Eq:RenyiDecomp}.
\end{proof}

\subsubsection{Proof of Lemma~\ref{Lem:RenyiRateLSI}}
\label{App:RenyiRateLSI}

We will use the following auxiliary results.
Recall that given a map $T \colon \R^n \to \R^n$ and a probability distribution $\rho$, the pushforward $T_\#\rho$ is the distribution of $T(x)$ when $x \sim \rho$.

\begin{lemma}\label{Lem:RenyiBij}
Let $T \colon \R^n \to \R^n$ be a differentiable bijective map.
For any probability distributions $\rho,\nu$, and for all $q > 0$,
\begin{align*}
R_{q,T_\#\nu}(T_\#\rho) = R_{q,\nu}(\rho).
\end{align*}
\end{lemma}
\begin{proof}
Let $\tilde \rho = T_\#\rho$ and $\tilde \nu = T_\#\nu$.
By the change of variable formula,
\begin{align*}
\rho(x) &= \det(\nabla T(x)) \, \tilde \rho(T(x)), \\
\nu(x) &= \det(\nabla T(x)) \, \tilde \nu(T(x)).
\end{align*}
Since $T$ is differentiable and bijective, $\det(\nabla T(x)) \neq 0$.
Therefore,
\begin{align*}
\frac{\tilde \rho(T(x))}{\tilde \nu(T(x))} = \frac{\rho(x)}{\nu(x)}.
\end{align*}
Now let $X \sim \nu$, so $T(X) \sim \tilde \nu$.
Then for all $q > 0$.
\begin{align*}
F_{q,\tilde\nu}(\tilde\rho) 
  = \E_{\tilde\nu}\left[\left(\frac{\tilde \rho}{\tilde \nu}\right)^q\right]
  = \E_{X \sim \nu}\left[\left(\frac{\tilde \rho(T(X))}{\tilde \nu(T(X))}\right)^q\right]
  = \E_{X \sim \nu}\left[\left(\frac{\rho(X)}{\nu(X)}\right)^q\right]
  = F_{q,\nu}(\rho).
\end{align*}
Suppose $q \neq 1$.
Taking logarithm on both sides and dividing by $q-1 \neq 0$ yields $R_{q,\tilde\nu}(\tilde\rho) = R_{q,\nu}(\rho)$, as desired.
The case $q=1$ follows from taking limit $q \to 1$, or by an analogous direct argument:
\begin{align*}
H_{\tilde \nu}(\tilde \rho) 
  = \E_{\tilde\nu}\left[\frac{\tilde \rho}{\tilde \nu}  \log \frac{\tilde \rho}{\tilde \nu}\right]
  = \E_{X \sim \nu}\left[\frac{\tilde \rho(T(X))}{\tilde \nu(T(X))} \log \frac{\tilde \rho(T(X))}{\tilde \nu(T(X))}\right]
  = \E_{X \sim \nu}\left[\frac{\rho(X)}{\nu(X)} \log \frac{\rho(X)}{\nu(X)}\right]
  = H_{\nu}(\rho).
\end{align*}
\end{proof}

We have the following standard result on how the LSI constant changes under a Lipschitz mapping.
We recall that $T \colon \R^n \to \R^n$ is $L$-Lipschitz if $\|T(x)-T(y)\| \le L\|x-y\|$ for all $x,y \in \R^n$.
For completeness, we provide the proof of Lemma~\ref{Lem:LSILipschitz} in Appendix~\ref{App:ProofLSILipschitz}.

\begin{lemma}\label{Lem:LSILipschitz}
Suppose a probability distribution $\nu$ satisfies LSI with constant $\alpha > 0$.
Let $T \colon \R^n \to \R^n$ be a differentiable $L$-Lipschitz map.
Then $\tilde \nu = T_\#\nu$ satisfies LSI with constant $\alpha/L^2$.
\end{lemma}

We also recall the following result on how the LSI constant changes along Gaussian convolution.
We provide the proof of Lemma~\ref{Lem:LSIGaussianConv} in Appendix~\ref{App:ProofLSIGaussianConv}.

\begin{lemma}\label{Lem:LSIGaussianConv}
Suppose a probability distribution $\nu$ satisfies LSI with constant $\alpha > 0$.
For $t > 0$, the probability distribution $\tilde \nu_t = \nu \ast \mathcal{N}(0,\,2tI)$ satisfies LSI with constant $\big(\frac{1}{\alpha}+2t \big)^{-1}$.
\end{lemma}

We now derive a formula for the decrease of R\'enyi divergence along simultaneous heat flow.
We note the resulting formula~\eqref{Eq:RenyiHeat} is similar to the formula~\eqref{Eq:RenyiLD} for the decrease of R\'enyi divergence along the Langevin dynamics.
A generalization of the following formula is also useful for analyzing a proximal sampling algorithm under isoperimetry~\cite{CCSW22}.

\begin{lemma}\label{Lem:RenyiHeat}
For any probability distributions $\rho_0,\nu_0$,
and for any  $t \ge 0$, let $\rho_t = \rho_0 \ast \N(0,2tI)$ and $\nu_t = \nu_0 \ast \N(0,2tI)$.
Then for all $q > 0$,
\begin{align}\label{Eq:RenyiHeat}
\frac{d}{dt} R_{q,\nu_t}(\rho_t) = -q\frac{G_{q,\nu_t}(\rho_t)}{F_{q,\nu_t}(\rho_t)}.
\end{align}
\end{lemma}
\begin{proof}
By definition, $\rho_t$ and $\nu_t$ evolve following the simultaneous heat flow:
\begin{align}\label{Eq:SHF}
\part{\rho_t}{t} = \Delta \rho_t, \qquad \part{\nu_t}{t} = \Delta \nu_t.
\end{align}
We will use the following identity for any smooth function $h \colon \R^n \to \R$,
\begin{align*} 
\Delta(h^q) = \nabla \cdot \left(qh^{q-1} \nabla h\right)
= q(q-1)h^{q-2} \|\nabla h\|^2 + qh^{q-1} \Delta h.
\end{align*}
We will also use the integration by parts formula~\eqref{Eq:IBP}.
Then along the simultaneous heat flow~\eqref{Eq:SHF},
\begin{align}
\frac{d}{dt} F_{q,\nu_t}(\rho_t)
 &= \frac{d}{dt} \int \frac{\rho_t^q}{\nu_t^{q-1}} \, dx \notag \\
 &= \int q \left( \frac{\rho_t}{\nu_t} \right)^{q-1} \part{\rho_t}{t} \, dx - \int (q-1) \left(\frac{\rho_t}{\nu_t}\right)^q \part{\nu_t}{t} \, dx \notag \\
 &= q \int \left( \frac{\rho_t}{\nu_t} \right)^{q-1} \Delta \rho_t \, dx - (q-1)  \int \left(\frac{\rho_t}{\nu_t}\right)^q \Delta \nu_t \, dx \notag \\
 &= q \int \Delta\left(\left( \frac{\rho_t}{\nu_t} \right)^{q-1}\right) \rho_t \, dx - (q-1)  \int \Delta \left(\left(\frac{\rho_t}{\nu_t}\right)^q \right) \nu_t \, dx \notag \\
 &= q \int \left((q-1)(q-2) \left( \frac{\rho_t}{\nu_t} \right)^{q-3} \left\|\nabla \frac{\rho_t}{\nu_t} \right\|^2 + (q-1) \left( \frac{\rho_t}{\nu_t} \right)^{q-2} \Delta \frac{\rho_t}{\nu_t} \right) \rho_t \, dx \notag \\
 &~~~~ - (q-1) \int \left(q(q-1) \left( \frac{\rho_t}{\nu_t} \right)^{q-2} \left\|\nabla \frac{\rho_t}{\nu_t} \right\|^2 + q \left( \frac{\rho_t}{\nu_t} \right)^{q-1} \Delta \frac{\rho_t}{\nu_t} \right) \nu_t \, dx \notag \\
 &= -q(q-1) \int \nu_t \left(\frac{\rho_t}{\nu_t} \right)^{q-2} \left\|\nabla \frac{\rho_t}{\nu_t} \right\|^2 \, dx \notag \\
 &= -q(q-1) G_{q,\nu_t}(\rho_t).  \label{Eq:RenyiHeatF}
\end{align}
Note that the identity~\eqref{Eq:RenyiHeatF} above is analogous to the identity~\eqref{Eq:Fdot} along the Langevin dynamics.
Therefore, for $q \neq 1$,
\begin{align*}
\frac{d}{dt} R_{q,\nu_t}{\rho_t}
  = \frac{1}{q-1} \frac{\frac{d}{dt} F_{q,\nu_t}(\rho_t)}{F_{q,\nu_t}(\rho_t)} = -q \frac{G_{q,\nu_t}(\rho_t)}{F_{q,\nu_t}(\rho_t)},
\end{align*}
as desired.

The case $q=1$ follows from taking limit $q \to 1$, or by an analogous direct calculation.
We will use the following identity for $h \colon \R^n \to \R_{> 0}$,
\begin{align*}
\Delta \log h = \nabla \cdot \left(\frac{\nabla h}{h} \right) = \frac{\Delta h}{h} - \|\nabla \log h\|^2.
\end{align*}
Then along the simultaneous heat flow~\eqref{Eq:SHF},
\begin{align*}
\frac{d}{dt} H_{\nu_t}(\rho_t)
  &= \frac{d}{dt} \int \rho_t \log \frac{\rho_t}{\nu_t} \, dx \\
  &= \int \part{\rho_t}{t} \log \frac{\rho_t}{\nu_t} \, dx + \int \rho_t  \frac{\nu_t}{\rho_t} \part{}{t} \left(\frac{\rho_t}{\nu_t}\right) \, dx \\
  &= \int \Delta \rho_t \,  \log \frac{\rho_t}{\nu_t} \, dx + \int \nu_t \left( \frac{1}{\nu_t} \part{\rho_t}{t} \, dx - \frac{\rho_t}{\nu_t^2} \part{\nu_t}{t} \right) \, dx \\
  &= \int \rho_t \, \Delta \log \frac{\rho_t}{\nu_t} \, dx - \int \frac{\rho_t}{\nu_t} \Delta \nu_t \, dx \\
  &= \int \rho_t \, \left(\frac{\nu_t}{\rho_t} \Delta \left( \frac{\rho_t}{\nu_t} \right) - \left\|\nabla \log \frac{\rho_t}{\nu_t}\right\|^2
  \right)  \, dx - \int \frac{\rho_t}{\nu_t} \Delta \nu_t \, dx \\
  &= -J_{\nu_t}(\rho_t),
\end{align*}
as desired. 
Note that this is also analogous to the identity~\eqref{Eq:HdotLD} along the Langevin dynamics.
\end{proof}

We are now ready to prove Lemma~\ref{Lem:RenyiRateLSI}.

\begin{proof}[Proof of Lemma~\ref{Lem:RenyiRateLSI}]
We will prove that along each step of ULA~\eqref{Eq:ULA} from $x_k \sim \rho_k$ to $x_{k+1} \sim \rho_{k+1}$, the R\'enyi divergence with respect to $\nu_\step$ decreases by a constant factor:
\begin{align}\label{Eq:RenyiDecrease}
R_{q,\nu_\step}(\rho_{k+1}) \le e^{-\frac{\beta \step}{q}} R_{q,\nu_\step}(\rho_k).
\end{align}
Iterating the bound above yields the desired claim~\eqref{Eq:RenyiRateLSI}.

We decompose each step of ULA~\eqref{Eq:ULA} into a sequence of two steps:
\begin{subequations}
\begin{align}
\tilde \rho_k &= (I - \step \nabla f)_\# \rho_k, \label{Eq:ULA-1} \\
\rho_{k+1} &= \tilde \rho_k \ast \N(0,2\step I).  \label{Eq:ULA-2}
\end{align}
\end{subequations}

In the first step~\eqref{Eq:ULA-1}, we apply a smooth deterministic map $T(x) = x-\step \nabla f(x)$.
Since $\nabla f$ is $L$-Lipschitz and $\step < \frac{1}{L}$, $T$ is a bijection.
Then by Lemma~\ref{Lem:RenyiBij},
\begin{align}\label{Eq:Renyi-Calc1}
R_{q,\nu_\step}(\rho_k) = R_{q,\tilde \nu_\step}(\tilde \rho_k)
\end{align}
where $\tilde \nu_\step = (I - \step \nabla f)_\# \nu_\step$.
Recall by Assumption~\ref{As:RenyiLSI} that $\nu_\step$ satisfies LSI with constant $\beta$.
Since the map $T(x) = x - \step \nabla f(x)$ is $(1+\step L)$-Lipschitz, by Lemma~\ref{Lem:LSILipschitz} we know that $\tilde \nu_\step$ satisfies LSI with constant $\frac{\beta}{(1+\step L)^2}$.

In the second step~\eqref{Eq:ULA-2}, we convolve with a Gaussian distribution, which is the result of evolving along the heat flow at time $\step$.
For $0 \le t \le \step$, let $\tilde \rho_{k,t} = \tilde \rho_k \ast \N(0,2t I)$ and $\tilde \nu_{\step,t} = \tilde \nu_\step \ast \N(0,2tI)$,
so $\tilde \rho_{k,\step} = \tilde \rho_{k+1}$ and $\tilde \nu_{\step,\step} = \nu_\step$.
By Lemma~\ref{Lem:RenyiHeat},
\begin{align*}
\frac{d}{dt} R_{q,\tilde \nu_{\step,t}}(\tilde \rho_{k,t}) = -q\frac{G_{q,\tilde\nu_{\step,t}}(\tilde\rho_{k,t})}{F_{q,\tilde\nu_{\step,t}}(\tilde\rho_{k,t})}.
\end{align*}
Since $\tilde \nu_\step$ satisfies LSI with constant $\frac{\beta}{(1+\step L)^2}$, by Lemma~\ref{Lem:LSIGaussianConv} we know that $\tilde \nu_{\step,t}$ satisfies LSI with constant $\big(\frac{(1+\step L)^2}{\beta}+2t\big)^{-1} \ge \big(\frac{(1+\step L)^2}{\beta}+2\step \big)^{-1}$ for $0 \le t \le \step$. 
In particular, since $\step \le \min\{\frac{1}{3L}, \frac{1}{9\beta}\}$, the LSI constant is $\big(\frac{(1+\step L)^2}{\beta}+2\step \big)^{-1}  \ge \big(\frac{16}{9\beta}+\frac{2}{9\beta} \big)^{-1} = \frac{\beta}{2}$.
Then by Lemma~\ref{Lem:RenyiLSI},
\begin{align*}
\frac{d}{dt} R_{q,\tilde \nu_{\step,t}}(\tilde \rho_{k,t}) = -q\frac{G_{q,\tilde\nu_{\step,t}}(\tilde\rho_{k,t})}{F_{q,\tilde\nu_{\step,t}}(\tilde\rho_{k,t})} \le -\frac{\beta}{q} R_{q,\tilde\nu_{\step,t}}(\tilde\rho_{\step,t}).
\end{align*}
Integrating over $0 \le t \le \step$ gives
\begin{align}\label{Eq:Renyi-Calc2}
R_{q,\nu_\step}(\rho_{k+1}) = R_{q,\tilde\nu_{\step,\step}}(\tilde\rho_{k,\step}) \le e^{-\frac{\beta \step}{q}} R_{q,\tilde \nu_\step}(\tilde\rho_k).
\end{align}
Combining~\eqref{Eq:Renyi-Calc1} and~\eqref{Eq:Renyi-Calc2} gives the desired inequality~\eqref{Eq:RenyiDecrease}.
\end{proof}

\subsubsection{Proof of Theorem~\ref{Thm:RenyiRate}}
\label{App:ProofThmRenyiRate}

\begin{proof}[Proof of Theorem~\ref{Thm:RenyiRate}]
This follows directly from Lemma~\ref{Lem:RenyiDecomp} and Lemma~\ref{Lem:RenyiRateLSI}.
\end{proof}

\subsection{Details for \S\ref{Sec:Poincare}: Poincar\'e inequality}
\label{App:Poincare}

\subsubsection{Proof of Lemma~\ref{Lem:RenyiPI}}
\label{App:RenyiPI}

\begin{proof}[Proof of Lemma~\ref{Lem:RenyiPI}]
We plug in $g^2 = \left(\frac{\rho}{\nu}\right)^q$ to Poincar\'e inequality~\eqref{Eq:PI} and use the monotonicity condition from Lemma~\ref{Lem:RenyiIncreasing} to obtain
\begin{align*}
\frac{q^2}{4\alpha} G_{q,\nu}(\rho)
  &\ge F_{q,\nu}(\rho) - F_{\frac{q}{2},\nu}(\rho)^2 \\
  &= e^{(q-1) R_{q,\nu}(\rho)} - e^{(q-2) R_{\frac{q}{2},\nu}(\rho)} \\
  &\ge e^{(q-1) R_{q,\nu}(\rho)} - e^{(q-2) R_{q,\nu}(\rho)} \\
  &= F_{q,\nu}(\rho) \left(1-e^{-R_{q,\nu}(\rho)}\right).
\end{align*}
Dividing both sides by $F_{q,\nu}(\rho)$ and rearranging yields the desired inequality.
\end{proof}

\subsubsection{Proof of Theorem~\ref{Thm:RenyiLP}}
\label{App:RenyiLP}

\begin{proof}[Proof of Theorem~\ref{Thm:RenyiLP}]
By Lemma~\ref{Lem:RenyiLD} and Lemma~\ref{Lem:RenyiPI},
\begin{align*}
\frac{d}{dt}  R_{q,\nu}(\rho_t) = -q \frac{G_{q,\nu}(\rho_t)}{F_{q,\nu}(\rho_t)} 
\le -\frac{4\alpha}{q} \left(1-e^{-R_{q,\nu}(\rho_t)}\right).
\end{align*}
We now consider two possibilities:
\begin{enumerate}
  \item If $R_{q,\nu}(\rho_0) \ge 1$, then as long as $R_{q,\nu}(\rho_t) \ge 1$, we have $1-e^{-R_{q,\nu}(\rho_t)} \ge 1-e^{-1} > \frac{1}{2}$, so $\frac{d}{dt} R_{q,\nu}(\rho_t) \le -\frac{2\alpha}{q}$, which implies $R_{q,\nu}(\rho_t) \le R_{q,\nu}(\rho_0) - \frac{2\alpha t}{q}$.
 \item If $R_{q,\nu}(\rho_0) \le 1$, then $R_{q,\nu}(\rho_t) \le 1$, and thus $\frac{1-e^{-R_{q,\nu}(\rho_t)}}{R_{q,\nu}(\rho_t)} \ge \frac{1}{1+R_{q,\nu}(\rho_t)} \ge \frac{1}{2}$.
Thus, in this case
$\frac{d}{dt} R_{q,\nu}(\rho_t) \le -\frac{2\alpha}{q} R_{q,\nu}(\rho_t)$,
and integrating gives
$R_{q,\nu}(\rho_t) \le e^{-\frac{2 \alpha t}{q}} R_{q,\nu}(\rho_0)$,
as desired.
\end{enumerate}
\end{proof}

\subsubsection{Proof of Lemma~\ref{Lem:RenyiRateP}}
\label{App:RenyiRateP}

We will use the following auxiliary results, which are analogous to Lemma~\ref{Lem:LSILipschitz} and Lemma~\ref{Lem:LSIGaussianConv}.
We provide the proof of Lemma~\ref{Lem:PLipschitz} in Appendix~\ref{App:ProofPLipschitz}, and the proof of Lemma~\ref{Lem:PGaussianConv} in Appendix~\ref{App:ProofPGaussianConv}.

\begin{lemma}\label{Lem:PLipschitz}
Suppose a probability distribution $\nu$ satisfies Poincar\'e inequality with constant $\alpha > 0$.
Let $T \colon \R^n \to \R^n$ be a differentiable $L$-Lipschitz map.
Then $\tilde \nu = T_\#\nu$ satisfies Poincar\'e inequality with constant $\alpha/L^2$.
\end{lemma}

\begin{lemma}\label{Lem:PGaussianConv}
Suppose a probability distribution $\nu$ satisfies Poincar\'e inequality with constant $\alpha > 0$.
For $t > 0$, the probability distribution $\tilde \nu_t = \nu \ast \mathcal{N}(0,\,2tI)$ satisfies Poincar\'e inequality with constant $\big(\frac{1}{\alpha}+2t \big)^{-1}$.
\end{lemma}

We are now ready to prove Lemma~\ref{Lem:RenyiRateP}.

\begin{proof}[Proof of Lemma~\ref{Lem:RenyiRateP}]
Following the proof of Lemma~\ref{Lem:RenyiRateLSI}, 
we decompose each step of ULA~\eqref{Eq:ULA} into two steps:
\begin{subequations}
\begin{align}
\tilde \rho_k &= (I - \step \nabla f)_\# \rho_k, \label{Eq:ULA-3} \\
\rho_{k+1} &= \tilde \rho_k \ast \N(0,2\step I).  \label{Eq:ULA-4}
\end{align}
\end{subequations}
The first step~\eqref{Eq:ULA-3} is a deterministic bijective map, so it preserves R\'enyi divergence by Lemma~\ref{Lem:RenyiBij}: 
$R_{q,\nu_\step}(\rho_k) = R_{q,\tilde \nu_\step}(\tilde \rho_k)$, where $\tilde \nu_\step = (I - \step \nabla f)_\# \nu_\step$.
Recall by Assumption~\ref{As:RenyiP} that $\nu_\step$ satisfies Poincar\'e inequality with constant $\beta$.
Since the map $T(x) = x - \step \nabla f(x)$ is $(1+\step L)$-Lipschitz, by Lemma~\ref{Lem:PLipschitz} we know that $\tilde \nu_\step$ satisfies Poincar\'e inequality with constant $\frac{\beta}{(1+\step L)^2}$.

The second step~\eqref{Eq:ULA-4} is convolution with a Gaussian distribution, which is the result of evolving along the heat flow at time $\step$.
For $0 \le t \le \step$, let $\tilde \rho_{k,t} = \tilde \rho_k \ast \N(0,2t I)$ and $\tilde \nu_{\step,t} = \tilde \nu_\step \ast \N(0,2tI)$,
so $\tilde \rho_{k,\step} = \tilde \rho_{k+1}$ and $\tilde \nu_{\step,\step} = \nu_\step$.
Since $\tilde \nu_\step$ satisfies Poincar\'e inequality with constant $\frac{\beta}{(1+\step L)^2}$, by Lemma~\ref{Lem:PGaussianConv} we know that $\tilde \nu_{\step,t}$ satisfies Poincar\'e inequality with constant $\big(\frac{(1+\step L)^2}{\beta}+2t\big)^{-1} \ge \big(\frac{(1+\step L)^2}{\beta}+2\step \big)^{-1}$ for $0 \le t \le \step$. 
In particular, since $\step \le \min\{\frac{1}{3L}, \frac{1}{9\beta}\}$, the Poincar\'e constant is $\big(\frac{(1+\step L)^2}{\beta}+2\step \big)^{-1}  \ge \big(\frac{16}{9\beta}+\frac{2}{9\beta} \big)^{-1} = \frac{\beta}{2}$.
Then by Lemma~\ref{Lem:RenyiHeat} and Lemma~\ref{Lem:RenyiPI}, 
\begin{align*}
\frac{d}{dt} R_{q,\tilde \nu_{\step,t}}(\tilde \rho_{k,t}) = -q\frac{G_{q,\tilde\nu_{\step,t}}(\tilde\rho_{k,t})}{F_{q,\tilde\nu_{\step,t}}(\tilde\rho_{k,t})} \le -\frac{2\beta}{q} \left(1-e^{-R_{q,\tilde\nu_{\step,t}}(\tilde\rho_{k,t})}\right).
\end{align*}
We now consider two possibilities, as in Theorem~\ref{Thm:RenyiLP}:
\begin{enumerate}
  \item If $R_{q,\nu_\step}(\rho_k) = R_{q,\tilde\nu_{\step,0}}(\tilde\rho_{k,0}) \ge 1$, then as long as $R_{q,\nu_\step}(\rho_{k+1}) = R_{q,\tilde\nu_{\step,\step}}(\tilde\rho_{k,\step}) \ge 1$, we have $1-e^{-R_{q,\tilde\nu_{\step,t}}(\tilde\rho_{k,t})} \ge 1-e^{-1} > \frac{1}{2}$, so $\frac{d}{dt} R_{q,\tilde \nu_{\step,t}}(\tilde \rho_{k,t}) \le -\frac{\beta}{q}$, which implies $R_{q,\nu_\step}(\rho_{k+1}) \le R_{q,\nu_\step}(\rho_{k}) - \frac{\beta \step}{q}$.
  Iterating this step, we have that $R_{q,\nu_\step}(\rho_{k}) \le R_{q,\nu_\step}(\rho_{0}) -  \frac{\beta \step k}{q}$ if $R_{q,\nu_\step}(\rho_0) \ge 1$ and as long as $R_{q,\nu_\step}(\rho_k) \ge 1$.
 \item If $R_{q,\nu_\step}(\rho_k) = R_{q,\tilde\nu_{\step,0}}(\tilde\rho_{k,0}) \le 1$, then $R_{q,\tilde\nu_{\step,t}}(\tilde\rho_{k,t}) \le 1$, and thus $\frac{1-e^{-R_{q,\tilde\nu_{\step,t}}(\tilde\rho_{k,t})}}{R_{q,\tilde\nu_{\step,t}}(\tilde\rho_{k,t})} \ge \frac{1}{1+R_{q,\tilde\nu_{\step,t}}(\tilde\rho_{k,t})} \ge \frac{1}{2}$.
Thus, in this case
$\frac{d}{dt} R_{q,\tilde\nu_{\step,t}}(\tilde\rho_{k,t}) \le -\frac{\beta}{q} R_{q,\tilde\nu_{\step,t}}(\tilde\rho_{k,t})$.
Integrating over $0 \le t \le \step$ gives
$R_{q,\nu_\step}(\rho_{k+1}) = R_{q,\tilde\nu_{\step,\step}}(\tilde\rho_{k,\step}) \le e^{-\frac{\beta \step}{q}} R_{q,\tilde\nu_{\step,0}}(\tilde\rho_{k,0}) = e^{-\frac{\beta \step}{q}} R_{q,\nu_\step}(\rho_k)$.
Iterating this step gives $R_{q,\nu_\step}(\rho_k) \le e^{-\frac{\beta \step k}{q}} R_{q,\nu_\step}(\rho_0)$ if $R_{q,\nu_\step}(\rho_0) \le 1$, as desired.
\end{enumerate}
\end{proof}

\subsubsection{Proof of Theorem~\ref{Thm:RenyiRatePoincare}}
\label{App:RenyiRatePoincare}

\begin{proof}[Proof of Theorem~\ref{Thm:RenyiRatePoincare}]

By Lemma~\ref{Lem:RenyiRateP} (which applies since $2q > 2$), after $k_0$ iterations we have $R_{2q,\nu_\step}(\rho_{k_0}) \le 1$.
Applying the second case of Lemma~\ref{Lem:RenyiRateP} starting from $k_0$ gives $R_{2q,\nu_\step}(\rho_k) \le e^{-\frac{\beta \step (k-k_0)}{2q}} R_{2q,\nu_\step}(\rho_{k_0}) \le e^{-\frac{\beta \step (k-k_0)}{2q}}$.
Then by Lemma~\ref{Lem:RenyiDecomp},
\begin{align*}
R_{q,\nu}(\rho_k) &\le \left(\frac{q-\frac{1}{2}}{q-1}\right) R_{2q,\nu_\step}(\rho_k) + R_{2q-1,\nu}(\nu_\step)
\le \left(\frac{q-\frac{1}{2}}{q-1}\right) e^{-\frac{\beta \step (k-k_0)}{2q}} + R_{2q-1,\nu}(\nu_\step)
\end{align*}
as desired.
\end{proof}

\subsection{Proofs for~\S\ref{Sec:Bias}: Properties of biased limit}

\subsubsection{Bounding relative Fisher information}

Let $H(\rho) = -\E_\rho[\log \rho]$ be Shannon entropy, $J(\rho) = \E_\rho[\|\nabla \log \rho\|^2]$ be the Fisher information, and $K(\rho) = \E_\rho[\|\nabla^2 \log \rho\|^2_{\HS}]$ be the second-order Fisher information.
We can write relative entropy as 
$$H_\nu(\rho) = \E_\rho[\log \frac{\rho}{\nu}] = -H(\rho) + \E_\rho[f]$$ 
and we can write relative Fisher information as
\begin{align}
    J_\nu(\nu_\step) = \E_{\nu_\step}\left[\left\| \nabla \log \frac{\nu_\step}{\nu} \right\|^2 \right]
    &= J(\nu_\step) + 2 \E_{\nu_\step}[\langle \nabla \log \nu_\step, \nabla f \rangle] + \E_{\nu_\step}[\|\nabla f\|^2] \notag \\
    &= J(\nu_\step) + \E_{\nu_\step}[ \|\nabla f\|^2 -2 \Delta f ] \label{Eq:JnuCalc}
\end{align}
where the last step follows from integration by parts.

We first prove the following, which only requires second-order smoothness.

\begin{lemma}\label{Lem:Fisher}
Assume $f$ is $L$-smooth ($-L I \preceq \nabla^2 f \preceq LI$) and $\step \le \frac{1}{2L}$.
Then
\begin{align*}
    J_\nu(\nu_\step) \le \E_{\nu_\step}[\|\nabla f\|^2 - \Delta f] + \step  \, nL^2.
\end{align*}
\end{lemma}
\begin{proof}
We examine how entropy changes from $\nu_\step$ to $\mu_\step$ and back, which will give us an estimate on the Fisher information.
By the change-of variable formula for $\mu_\step = (I - \step \nabla f)_\# \nu_\step$, we have
\begin{align}\label{Eq:ChangeVar}
    \log \nu_\step(x) = \log \det (I - \step \nabla^2 f(x)  + \log \mu_\step(x - \step \nabla f(x)).
\end{align}
By taking expectation over $x \sim \nu$ (equivalently, $x - \step \nabla f(x) \sim \mu$), we get
\begin{align}\label{Eq:BCalc1}
    H(\nu_\step) = H(\mu_\step) - \E_{\nu_\step}[\log \det (I - \step \nabla^2 f)].
\end{align}

On the other hand, recall that along the heat flow $\rho_t = \rho_0 \ast \N(0,2tI)$, we have the relations
\begin{align*}
    \frac{d}{dt} H(\rho_t) &= J(\rho_t), \\
    \frac{d}{dt} J(\rho_t) &= -K(\rho_t) \le 0.
\end{align*}
See for example~\cite{V00}.
Thus, $\nu_\step = \mu_\step \ast \N(0, 2\step I)$ satisfies
\begin{align}\label{Eq:BCalc15}
    H(\nu_\step) = H(\mu_\step) + \int_0^\step J(\rho_t) \, dt \ge H(\mu_\step) + \step J(\nu_\step)
\end{align}
where $\rho_t = \rho_0 \ast \N(0,2tI)$ is the heat flow from $\rho_0 = \mu_\step$ to $\rho_\step = \nu_\step$, and the last inequality holds since $t \mapsto J(\rho_t)$ is decreasing.
Combining~\eqref{Eq:BCalc1} and~\eqref{Eq:BCalc15}, we get
\begin{align}\label{Eq:BCalc2}
  \step J(\nu_\step) 
  \le 
  H(\nu_\step) - H(\mu_\step) 
  = -\E_{\nu_\step}[\log \det (I - \step \nabla^2 f)]. 
\end{align}
Let $\lambda_1, \dots, \lambda_d$ be the eigenvalues of $\nabla^2 f$.
Since $f$ is $L$-smooth, $|\lambda_i| \le L$.
Using the inequality $\log(1-\step \lambda_i) \ge -\step \lambda_i - \step^2 \lambda_i^2$, which holds since $\step |\lambda_i| \le \frac{1}{2}$ since $\step \le \frac{1}{2L}$, we have
\begin{align*}
    -\E_{\nu_\step}[\log \det (I - \step \nabla^2 f)]
    &= \sum_{i=1}^n \E_{\nu_\step}[-\log (1-\step \lambda_i)]  \\
    &\le \sum_{i=1}^n \E_{\nu_\step}[\step \lambda_i + \step^2 \lambda_i^2] \\
    &= \step  \E_{\nu_\step}[\Delta f] + \step^2 \, \E_{\nu_\step}[\|\nabla^2 f\|^2_{\HS}] \\
    &\le  \step  \E_{\nu_\step}[\Delta f] + \step^2 \, n L^2.
\end{align*}
Plugging this to~\eqref{Eq:BCalc2} gives
\begin{align}\label{Eq:BCalc3}
    J(\nu_\step) 
    &\le -\frac{1}{\step}
      \E_{\nu_\step}[\log \det (I - \step \nabla^2 f)]
    \le \E_{\nu_\step}[\Delta f] + \step \, n L^2.
\end{align}
Therefore, we can bound the relative Fisher information~\eqref{Eq:JnuCalc}:
\begin{align*}
    J_\nu(\nu_\step)
    &= J(\nu_\step) + \E_{\nu_\step}[\|\nabla f\|^2 -2 \Delta f ]
    \le \E_{\nu_\step}[\|\nabla f\|^2 - \Delta f] + \step  \, nL^2.
\end{align*}
\end{proof}

\subsubsection{Bounding the expected value}

Recall that for $\nu = e^{-f}$, we have $\E_{\nu}[\|\nabla f\|^2 - \Delta f] = 0$.
Under third-order smoothness, we will prove $\E_{\nu_\step}[\|\nabla f\|^2 - \Delta f] = O(\step)$.

\begin{lemma}\label{Lem:Exp}
Assume $f$ is $(L,M)$-smooth.
Then for $\step \le \frac{1}{2L}$,
\begin{align}
    \E_{\nu_\step}[\|\nabla f\|^2 - \Delta f] \le  \step n \left(L^2 + M \sqrt{J(\mu_\step)}\right).
\end{align}
\end{lemma}
\begin{proof}
We examine how the expected value of $f$ changes from $\nu_\step$ to $\mu_\step$ and back, which will give us an estimate on the desired quantity.

Let $x \sim \nu_\step$ and $y = x - \step \nabla f(x) \sim \mu_\step$, so $x' = y + \sqrt{2\step} \, Z \sim \nu_\step$ where $Z \sim \N(0,I)$ is independent.
Since $f$ is $L$-smooth, we have the bound:
\begin{align*}
    f(y) \le f(x) - \step\left(1-\frac{\step L}{2}\right) \|\nabla f(x)\|^2.
\end{align*}
Taking expectation over $x \sim \nu_\step$ (equivalently, $y \sim \mu_\step$) yields
\begin{align*}
    \E_{\mu_\step}[f] \le \E_{\nu_\step}[f] - \step\left(1-\frac{\step L}{2}\right) \E_{\nu_\step}[\|\nabla f\|^2].
\end{align*}

On the other hand, let $\rho_t = \rho_0 \ast \N(0,2t I)$ be the heat flow from $\rho_0 = \mu_\step$ to $\rho_\step = \nu_\step$, and recall that along the heat flow, $\frac{d}{dt} \E_{\rho_t}[f] = \E_{\rho_t}[\Delta f]$. 
Then
\begin{align*}
    \E_{\nu_\step}[f] = \E_{\mu_\step}[f] + \int_0^\step \E_{\rho_t}[\Delta f] \, dt.
\end{align*}
Combining the two relations above,
\begin{align*}
    \step\left(1-\frac{\step L}{2}\right) \E_{\nu_\step}[\|\nabla f(x)\|^2] \le \E_{\nu_\step}[f] - \E_{\mu_\step}[f] = \int_0^\step \E_{\rho_t}[\Delta f] \, dt.
\end{align*}
Therefore,
\begin{align*}
    \step\left(1-\frac{\step L}{2}\right) \E_{\nu_\step}[\|\nabla f\|^2 - \Delta f] 
    &\le \int_0^\step \E_{\rho_t}[\Delta f] \, dt - \step\left(1-\frac{\step L}{2}\right) \E_{\nu_\step}[\Delta f] \\
    &= \int_0^\step (\E_{\rho_t}[\Delta f] - \E_{\rho_\step}[\Delta f]) \, dt + \frac{\step^2 L}{2} \E_{\nu_\step}[\Delta f] \\
    &\le \int_0^\step (\E_{\rho_t}[\Delta f] - \E_{\rho_\step}[\Delta f]) \, dt + \frac{\step^2 n L^2}{2}.\label{Eq:CalcB}
\end{align*}
Since $\rho_t$ evolves following the heat flow, by Lemma~\ref{Lem:Heat} we have for any $0 \le t \le \step$:
\begin{align*}
    W_2(\rho_t,\rho_\step)^2 &\le (\step-t)^2 J(\rho_t) \le (\step-t)^2 \, J(\rho_0) = (\step-t)^2 \, J(\mu_\step)
\end{align*}
where the second inequality above follows from the fact that Fisher information is decreasing along heat flow.

Since we assume $\nabla^2 f$ is $M$-Lipschitz, the Laplacian $\Delta f = \Tr(\nabla^2 f)$ is $(n M)$-Lipschitz.
Then by the dual formulation of $W_1$ distance,\footnote{Recall $W_1(\rho,\nu) = \sup\{ \E_\rho[g] - \E_\nu[g] \colon g ~ \text{ is $1$-Lipschitz}\}$.}
\begin{align*}
    \E_{\rho_t}[\Delta f] - \E_{\rho_\step}[\Delta f] &\le nM \, W_1(\rho_t, \rho_\step)
    \le nM \, W_2(\rho_t, \rho_\step)
    \le (\step-t) \, nM \sqrt{J(\mu_\step)}.
\end{align*}
Integrating over $0 \le t \le \step$ gives
\begin{align*}
    \int_0^\step (\E_{\rho_t}[\Delta f] - \E_{\rho_\step}[\Delta f]) \, dt \le \frac{\step^2}{2} nM \sqrt{J(\mu_\step)}.
\end{align*}
Plugging this to~\eqref{Eq:CalcB} gives
\begin{align*}
    \step\left(1-\frac{\step L}{2}\right) \E_{\nu_\step}[\|\nabla f\|^2 - \Delta f] 
    &\le \frac{\step^2}{2} nM \sqrt{J(\mu_\step)} + \frac{\step^2 n L^2}{2}.
\end{align*}
Since $1-\frac{\step L}{2} \ge \frac{3}{4} > \frac{1}{2}$ for $\step \le \frac{1}{2L}$, this also implies
\begin{align*}
    \frac{\step}{2} \E_{\nu_\step}[\|\nabla f\|^2 - \Delta f] 
    &\le \frac{\step^2}{2} nM \sqrt{J(\mu_\step)} + \frac{\step^2 n L^2}{2}.
\end{align*}
Dividing by $\frac{\step}{2}$ gives the claim.
\end{proof}

\begin{remark}
We see from~\eqref{Eq:CalcB} that if $\Delta \Delta f \ge 0$, then $\E_{\rho_t}[\Delta f] - \E_{\rho_\step}[\Delta f] \le 0$. Thus, we also get the bound $J_\nu(\nu_\step) \le \step n L^2$ assuming $f$ is $L$-smooth and $\Delta \Delta f \ge 0$.
\end{remark}

In the proof above, we use the following lemma on the distance along the heat flow.
Note that a simple coupling argument gives $W_2(\rho_\step, \rho_0)^2 \le O(\step)$, rather than $O(\eta^2)$ below (when $J(\rho_0) < \infty$).

\begin{lemma}\label{Lem:Heat}
For any probability distribution $\rho_0$ and for any $\step > 0$, let $\rho_\step = \rho_0 \ast \N(0, 2\step I)$.
Then
\begin{align}
    W_2(\rho_\step, \rho_0)^2 \le \step^2 J(\rho_0).
\end{align}
\end{lemma}
\begin{proof}
By definition, $\rho_\step$ evolves following the heat flow $\part{\rho_t}{t} = \Delta \rho_t$ from time $t=0$ to time $t=\step$.
Fix $\step > 0$, and let us rescale time to be from $0$ to $1$:
Let $\tilde \rho_\tau = \rho_{\tau \step}$, so $\tilde \rho_0 = \rho_0$ and $\tilde \rho_1 = \rho_\step$.
Then $\tilde \rho_\tau$ evolves following a rescaled heat flow:
\begin{align}
    \part{\tilde \rho_\tau}{\tau} = \part{\rho_{\tau \step}}{\tau} = \step \Delta \rho_{\tau \step} = \step \Delta \tilde \rho_\tau
    = \step \nabla \cdot \left( \tilde \rho_\tau  \nabla \log \tilde \rho_\tau \right)
\end{align}
Since $(\tilde \rho_\tau)_{0 \le \tau \le 1}$ connects $\tilde \rho_0 = \rho_0$ to $\tilde \rho_1 = \rho_{\step}$, its length must exceed the $W_2$ distance: 
\begin{align*}
    W_2(\rho_\step,\rho_0)^2 &\le \int_0^1 \E_{ \tilde \rho_\tau}[\|\step \, \nabla \log \tilde \rho_\tau\|^2] \, d\tau
    = \step^2 \int_0^1 J(\tilde \rho_\tau) \, d\tau
    \le \step^2 J(\rho_0).
\end{align*}
In the last step we have used Fisher information is decreasing along the heat flow: $J(\tilde \rho_\tau) \le J(\rho_0)$.
\end{proof}

\subsubsection{Bounding the Fisher information}

\begin{lemma}\label{Lem:FisherHalf}
\begin{enumerate}
\item If $f$ is $L$-smooth and $\step \le \frac{1}{2L}$, then
\begin{align*}
    J(\nu_\step) \le \frac{3}{2} nL.
\end{align*}
\item If $f$ is $(L,M)$-smooth and $\step \le \frac{1}{2L}$, then
\begin{align*}
    J(\mu_\step) \le 12n (L + 3nM^2).
\end{align*}
\end{enumerate}
\end{lemma}
\begin{proof}
First, since $f$ is $L$-smooth and $\step \le \frac{1}{2L}$, from~\eqref{Eq:BCalc3} we can bound
\begin{align}
    J(\nu_\step) \le \E_{\nu_\step}[\Delta f] + \step \, \E_{\nu_\step}[\|\nabla^2 f\|^2_{\HS}] \le nL + \step nL^2 \le \frac{3}{2} nL.
\end{align}

Second, by taking gradient in the formula~\eqref{Eq:ChangeVar} for $\mu_\step = (I - \step \nabla f)_\# \nu_\step$, we get
\begin{align*}
    \nabla \log \nu_\step(x) = -\step \nabla^3 f(x) \, A(x)^{-1} + A(x) \nabla \log \mu_\step(x - \step \nabla f(x))
\end{align*}
or equivalently,
\begin{align}\label{Eq:CalcA}
    \nabla \log \mu_\step(x - \step \nabla f(x))
    = A(x)^{-1} \nabla \log \nu_\step(x) + \step A(x)^{-1} \nabla^3 f(x) \, A(x)^{-1}
\end{align}
where
\begin{align*}
    A(x) = I - \step \nabla^2 f(x).
\end{align*}
satisfies $ \frac{1}{2} I \preceq\, A(x) \,\preceq\, \frac{3}{2} I$ since $-L I \preceq \nabla^2 f(x) \preceq L I$ and $\step \le \frac{1}{2L}$.
In particular,
\begin{align*}
    \frac{2}{3} I \,\preceq\, A(x)^{-1} \,\preceq\, 2 I.
\end{align*}
Therefore, the first term in~\eqref{Eq:CalcA} we can bound as
\begin{align*}
    \|A(x)^{-1} \nabla \log \nu_\step(x)\|_2 &\le 2 \|\nabla \log \nu_\step(x)\|_2.
\end{align*}
For the second term, using the assumption $\|\nabla^3 f(x)\|_\op \le M$ and Lemma~\ref{Lem:Smooth}, we have
\begin{align*}
    \|A(x)^{-1} \nabla^3 f(x) \, A(x)^{-1}\|_2 &\le 2 \|\nabla^3 f(x) \, A(x)^{-1}\|_2 
    \,\le\, 4n M.
\end{align*}
Therefore, from~\eqref{Eq:CalcA}, we get
\begin{align*}
    \|\nabla \log \mu_\step(x - \step \nabla f(x))\|_2 \le 2 \|\nabla \log \nu_\step(x)\|_2 + 4n M.
\end{align*}
This implies
\begin{align*}
    \|\nabla \log \mu_\step(x - \step \nabla f(x))\|_2^2 \le 8 \|\nabla \log \nu_\step(x)\|_2^2 + 32n^2 M^2.
\end{align*}
Taking expectation over $x \sim \nu_\step$ (equivalently, $x - \step \nabla f(x) \sim \mu_\step$), we conclude that
\begin{align*}
    J(\mu_\step) &\le 8 J(\nu_\step) + 32 n^2 M^2
    \le 12 nL + 32 n^2 M^2
    \le 12n (L + 3nM^2).
\end{align*}
\end{proof}

In the above, we use the following bound from smoothness.

\begin{lemma}\label{Lem:Smooth}
Let $T \in \R^{d \times d \times d}$ be a $3$-tensor with $\|T\|_\op \le M$.
For any symmetric matrix $B \in \R^{d \times d}$ with $\|B\|_\op \le \beta$, the vector $TB \in \R^n$ satisfies $\| TB \|_2 \le n \beta M$.
\end{lemma}
\begin{proof}
Since $\|T\|_\op \le M$, for any $u,v,w \in \R^d$ with $\|u\| = \|v\| = \|w\| = 1$, $|T [u,v,w]| \le M$.
In particular, for any $u \in \R^d$ with $\|u\| = 1$, $p = T[u,u] \in \R^d$ satisfies $\|p\|_2 \le M$.
We eigendecompose $B = \sum_{i=1}^d \lambda_i u_i u_i^\top$ with eigenvectors $u_1,\dots,u_n \in \R^n$ and eigenvalues $\lambda_1,\dots,\lambda_n \in \R$ with $\|u_i\|_2 = 1$, $|\lambda_i| \le \beta$.
Then
\begin{align*}
    \| TB\|_2 &= \left\| T \sum_{i=1}^n \lambda_i u_i u_i^\top \right\|_2 
    \le \sum_{i=1}^n |\lambda_i| \cdot \|T[u_i,u_i]\|_2 
    \le \sum_{i=1}^n \beta M
    = n \beta M.
\end{align*}
\end{proof}

\subsubsection{Proof of  upper bound in Theorem~\ref{Thm:Bias}}
\label{Sec:ProofBiasUpper}

\begin{proof}[Proof of  upper bound in Theorem~\ref{Thm:Bias}]
By combining Lemmas~\ref{Lem:Fisher},~\ref{Lem:Exp}, and~\ref{Lem:FisherHalf}:
\begin{align*}
    J_\nu(\nu_\step) &\le \E_{\nu_\step}[\|\nabla f\|^2 - \Delta f] + \step  nL^2 \\
    &\le \step n \left(2L^2 + M \sqrt{J(\mu_\step)}\right) \\
    &\le \step n \left(2L^2 + M \sqrt{12n (L + 3nM^2)}\right) \\
    &\le \step n \left(2L^2 + M (4 \sqrt{nL} + 6nM)\right) \\
    &\le 2 \step n \left(L^2 + 2 \sqrt{nL}M + 3nM^2\right).
\end{align*}

\end{proof}

\subsubsection{Proof of lower bound in Theorem~\ref{Thm:Bias}}
\label{Sec:ProofBiasLower}

For the lower bound, we first prove the following properties.
Observe that for $\nu \propto e^{-f}$, $\E_\nu[\nabla f] = 0$ and $\E_{\nu}[\langle x, \nabla f(x)\rangle ] = n$.
These properties still hold when we take the expectation under the biased limit.

\begin{lemma}\label{Lem:Biased}
    For any $f$ and $\step > 0$, the biased limit $\nu_\step$ satisfies:
\begin{enumerate}
    \item $\E_{\nu_\step}[\nabla f] = 0$.
    
    \item $\E_{\nu_\step}[\langle x, \nabla f(x)\rangle ] = d + \frac{\step}{2} \E_{\nu_\step}[\|\nabla f\|^2]$.
\end{enumerate}
\end{lemma}

\begin{proof}
Let $x \sim \nu_\step$, $y = x - \step \nabla f(x) \sim \mu_\step$, and $x' = y + \sqrt{2\step} z \sim \nu_\step$ where $z \sim \N(0,I)$ is independent.
Then
\begin{align*}
    x' = x - \step \nabla f(x) + \sqrt{2\step} z.
\end{align*}
By taking expectation over $x \sim \nu_\step$ (so $x' \sim \nu_\step$), we get:
\begin{align*}
    \E_{\nu_\step}[x'] = \E_{\nu_\step}[x] - \E_{\nu_\step}[\nabla f(x)]
\end{align*}
which implies $\E_{\nu_\step}[\nabla f] = 0$.

Next, by taking covariance, we get:
\begin{align*}
    \Cov_{\nu_\step}(x') &= \Cov_{\nu_\step}(x - \step \nabla f(x)) + 2\step I \\   
    &= \Cov_{\nu_\step}(x) - \step \Cov_{\nu_\step}(x, \nabla f(x)) - \step \Cov_{\nu_\step}(\nabla f(x),x) + \step^2 \Cov_{\nu_\step}(\nabla f(x)) + 2\step I
\end{align*}
so
\begin{align*}
  \Cov_{\nu_\step}(x, \nabla f(x)) + \Cov_{\nu_\step}(\nabla f(x),x) 
  &= \step \Cov_{\nu_\step}(\nabla f(x)) + 2 I.
\end{align*}
Since $\E_{\nu_\step}[\nabla f] = 0$, this means
\begin{align*}
  \E_{\nu_\step}[x \, \nabla f(x)^\top ] + \E_{\nu_\step}[\nabla f(x) \, x^\top] 
  &= \step \E_{\nu_\step}[\nabla f(x) \, \nabla f(x)^\top] + 2 I.
\end{align*}
Taking trace and dividing by $2$ gives
\begin{align*}
    \E_{\nu_\step}[\langle x, \nabla f(x)\rangle ] = n + \frac{\step}{2} \E_{\nu_\step}[\|\nabla f(x)\|^2] .
\end{align*}
\end{proof}

We are now ready to prove the lower bound.

\begin{proof}[Proof of lower bound in Theorem~\ref{Thm:Bias}]
From Lemma~\ref{Lem:Biased} part 2, using the identity $d = \E_{\nu_\step}[\langle x, -\nabla \log \nu_\step \rangle]$ and Cauchy-Schwarz inequality, we can derive the bound:
\begin{align*}
    \frac{\step}{2} \E_{\nu_\step}[\|\nabla f(x)\|^2]
    &= \E_{\nu_\step}[\langle x, \nabla f(x)\rangle ] - d \\
    &= \E_{\nu_\step}\left[\left\langle x, \nabla \log \frac{\nu_\step}{\nu}\right\rangle \right] \\
    &= \E_{\nu_\step}\left[\left\langle x - \E_{\nu_\step}[x], \nabla \log \frac{\nu_\step}{\nu}\right\rangle \right] \\
    &\le \sqrt{\Var_{\nu_\step}(x)} \cdot \sqrt{J_\nu(\nu_\step)}.
\end{align*}
Rearranging gives us the desired result.
In the second step above we can subtract $\E_{\nu_\step}[x]$ because for any $c \in \R^n$, by Lemma~\ref{Lem:Biased} part 1,
\begin{align*}
    \E_{\nu_\step}\left[\left\langle c, \nabla \log \frac{\nu_\step}{\nu} \right\rangle\right]
    &= \langle c, \E_{\nu_\step}[\nabla \log \nu_\step] + \E_{\nu_\step}[\nabla f] \rangle
    = 0.
\end{align*}
\end{proof}

\subsubsection{Proof of Theorem~\ref{Thm:BiasLSI}}
\label{Sec:ProofBiasLSI}

The following proof is due to Sinho Chewi.

\begin{proof}[Proof of Theorem~\ref{Thm:BiasLSI}]
Suppose we run ULA from $x_0 \sim \rho_0$ to obtain $x_k \sim \rho_k$, so $\rho_k \to \nu_\step$ as $k \to \infty$.
Let $\alpha_k$ denote the LSI constant of $\rho_k$, i.e.\ the largest constant $\tilde \alpha > 0$ such that~\eqref{Eq:LSI-KL} holds.
Since $0 \le \step \le \frac{1}{L}$ and $f$ is $\alpha$-strongly convex, the map $x \mapsto x - \step \nabla f(x)$ is $(1-\step \alpha)$-Lipschitz.
Since $x_k \sim \rho_k$ is $\alpha_k$-LSI, by Lemma~\ref{Lem:LSILipschitz}, the distribution of $x_k - \step \nabla f(x_k)$ satisfies LSI with constant $\alpha_k / (1-\step \alpha)^2$.
Then by Lemma~\ref{Lem:LSIGaussianConv}, $x_{k+1} = x_k - \step \nabla f(x_k) + \sqrt{2\step} \, z_k \sim \rho_{k+1}$ satisfies $\alpha_{k+1}$-LSI with
\begin{align*} 
    \frac{1}{\alpha_{k+1}} \le \frac{(1-\step \alpha)^2}{\alpha_k} + 2\step.
\end{align*}
Suppose we start $\alpha_0 \ge \frac{\alpha}{2}$.
We claim that $\alpha_k \ge \frac{\alpha}{2}$ for all $k \ge 0$.
Indeed, if $\frac{1}{\alpha_k} \le \frac{2}{\alpha}$, then since $\step \le \frac{1}{L} \le \frac{1}{\alpha}$, we have
\begin{align*}
    \frac{1}{\alpha_{k+1}} &\le \frac{(1-\step \alpha)^2}{\alpha/2} + 2\step
    = \frac{2}{\alpha} - 2\step (1-\step \alpha)
    \le \frac{2}{\alpha}.
\end{align*}
Thus by induction, $\alpha_k \ge \frac{\alpha}{2}$ for all $k \ge 0$.
Taking the limit $k \to \infty$, this shows that $\nu_\step = \lim_{k \to \infty} \rho_k$ also satisfies LSI with constant $\beta \ge \frac{\alpha}{2}$.
\end{proof}

\section{Discussion}
\label{Sec:Disc}

In this paper we proved convergence guarantees on KL divergence and R\'enyi divergence along ULA under isoperimetric assumptions and bounded Hessian, without assuming convexity or bounds on higher derivatives.
In particular, under LSI and bounded Hessian, we prove a complexity guarantee of $O(\frac{\kappa^2 n}{\err})$ to achieve $H_\nu(\rho_k) \le \err$, where $\kappa := L/\alpha$ is the condition number. 
We note the dependence on $\kappa$ may not be tight.
In particular, the asymptotic bias in KL divergence from our result scales linearly with step size, while from the Gaussian example we see it should scale quadratically with step size.
We can achieve a smaller bias using a different algorithm, e.g.\ the underdamped Langevin algorithm~\cite{Ma19} or the proximal Langevin algorithm~\cite{W19}.
However, it remains open whether we can provide a better analysis of ULA under LSI and smoothness that yields the optimal bias.

Our convergence results for ULA in R\'enyi divergence hold assuming the biased limit satisfies isoperimetry (Assumptions~\ref{As:RenyiLSI} and~\ref{As:RenyiP}), which we can verify assuming strong log-concavity and smoothness of the target distribution.
It would be interesting to verify when Assumptions~\ref{As:RenyiLSI} and~\ref{As:RenyiP} hold more generally, whether they can be relaxed, or if they follow from assuming isoperimetry and smoothness for the target density.

Another intriguing question is whether there is an affine-invariant version of the Langevin dynamics. This might lead to a sampling algorithm with logarithmic dependence on smoothness parameters, rather than the current polynomial dependence.
There are some approaches that achieve affine invariance in continuous time, for example via interacting Langevin dynamics~\cite{GNR20} or the Newton Langevin dynamics~\cite{CLLMRS20};
however, the discretization analysis remains a challenge.

Since the publication of the conference version of this paper~\cite{VW19}, some of our techniques and results have been generalized.
The one-step interpolation technique that we use in Lemma~\ref{Lem:OneStep} proves to be useful for analyzing ULA or its variants under various assumptions.
It has been extended to analyze ULA for sampling on manifolds, for example, on a complete Riemannian manifold~\cite{WLP20}; or on a product of spheres with applications to solving semidefinite programming~\cite{LE20}.
It has also been used to analyze ULA for sampling from distributions with sub-Gaussian tail growth and H\"older-continuous gradient~\cite{EH21}; for sampling from heavy-tailed distributions~\cite{HBE22}; 
for sampling in R\'enyi divergence under a family of isoperimetric inequalities interpolating between LSI and Poincar\'e~\cite{CELSZ21}; and for sampling from non-log-concave distributions with convergence in Fisher information~\cite{BCESZ22}.
The interpolation technique has also been useful for analyzing other sampling algorithms, e.g.\ the Proximal Langevin Algorithm (which uses the proximal method for $f$ rather than gradient descent), which yields a smaller (and tight) asymptotic bias~\cite{W19}.
It has also been used to analyze the Mirror Langevin Algorithm (which uses Hessian metric and discretizes in the dual space) for sampling under mirror isoperimetric inequalities~\cite{J20}. 
Further, the calculation along simultaneous heat flow (Lemma~\ref{Lem:RenyiHeat}) is also useful for analyzing the convergence of a new proximal sampler algorithm under isoperimetry~\cite{CCSW22}.


\paragraph{Acknowledgment.} The authors thank Kunal Talwar for explaining the privacy motivation and application of R\'enyi divergence to data privacy; Yu Cao, Jianfeng Lu, and Yulong Lu for alerting us to their work~\cite{CLL18} on R\'enyi divergence; Xiang Cheng and Peter Bartlett for helpful comments on an earlier version of this paper; and Sinho Chewi for communicating Theorem~\ref{Thm:BiasLSI} to us.

{\small
\bibliographystyle{plain}
\bibliography{main.bbl}
}

\appendix

\section{Appendix}

\subsection{Review on notation and basic properties}
\label{App:Notation}

Throughout, we represent a probability distribution $\rho$ on $\R^n$ via its probability density function with respect to the Lebesgue measure, so $\rho \colon \R^n \to \R$ with $\int_{\R^n} \rho(x) dx = 1$.
We typically assume $\rho$ has full support and smooth density, so $\rho(x) > 0$ and $x \mapsto \rho(x)$ is differentiable.
Given a function $f \colon \R^n \to \R$, we denote the expected value of $f$ under $\rho$ by
\begin{align*}
\E_\rho[f] = \int_{\R^n} f(x) \rho(x) \, dx.
\end{align*}
We use the Euclidean inner product $\langle x,y \rangle = \sum_{i=1}^n x_i y_i$ for $x = (x_i)_{1 \le i \le n}, y = (y_i)_{1 \le i \le n} \in \R^n$.
For symmetric matrices $A, B \in \R^{n \times n}$, let $A \preceq B$ denote that $B-A$ is positive semidefinite.
For $\mu \in \R^n$, $\Sigma \succ 0$, let $\N(\mu,\Sigma)$ denote the Gaussian distribution on $\R^n$ with mean $\mu$ and covariance matrix $\Sigma$.

Given a smooth function $f \colon \R^n \to \R$, its {\bf gradient} $\nabla f \colon \R^n \to \R^n$ is the vector of partial derivatives:
\begin{align*}
\nabla f(x) = \left(\part{f(x)}{x_1}, \dots, \part{f(x)}{x_n} \right).
\end{align*}
The {\bf Hessian} $\nabla^2 f \colon \R^n \to \R^{n \times n}$ is the matrix of second partial derivatives:
\begin{align*}
\nabla^2 f(x) = \left(\part{^2f(x)}{x_i x_j} \right)_{1 \le i,j \le n}.
\end{align*}
The {\bf Laplacian} $\Delta f \colon \R^n \to \R$ is the trace of its Hessian:
\begin{align*}
\Delta f(x) = \Tr(\nabla^2 f(x)) = \sum_{i=1}^n \part{^2 f(x)}{x_i^2}.
\end{align*}

Given a smooth vector field $v = (v_1,\dots,v_n) \colon \R^n \to \R^n$, its {\bf divergence} $\nabla \cdot v \colon \R^n \to \R$ is 
\begin{align*}
(\nabla \cdot v)(x) = \sum_{i=1}^n \part{v_i(x)}{x_i}.
\end{align*}
In particular, the divergence of gradient is the Laplacian:
\begin{align*}
(\nabla \cdot \nabla f)(x) = \sum_{i=1}^n \part{^2 f(x)}{x_i^2} = \Delta f(x).
\end{align*}

For any function $f \colon \R^n \to \R$ and vector field $v \colon \R^n \to \R^n$ with sufficiently fast decay at infinity,
we have the following {\bf integration by parts} formula:
\begin{align*}
\int_{\R^n} \langle v(x), \nabla f(x) \rangle dx = -\int_{\R^n} f(x) (\nabla \cdot v)(x) dx.
\end{align*}
Furthermore, for any two functions $f, g \colon \R^n \to \R$,
\begin{align*}
\int_{\R^n} f(x) \Delta g(x) dx = -\int_{\R^n} \langle \nabla f(x), \nabla g(x) \rangle dx = \int_{\R^n} g(x) \Delta f(x) dx.
\end{align*}
When the argument is clear, we omit the argument $(x)$ in the formulae for brevity.
For example, the last integral above becomes
\begin{align}\label{Eq:IBP}
\int f \, \Delta g \, dx = -\int \langle \nabla f, \nabla g \rangle \, dx = \int g \, \Delta f \, dx.
\end{align}

\subsection{Derivation of the Fokker-Planck equation}
\label{App:FP}

Consider a stochastic differential equation
\begin{align}\label{Eq:SDEex}
dX_t = v(X_t) \, dt + \sqrt{2} \, dW_t
\end{align}
where $v \colon \R^n \to \R^n$ is a smooth vector field and $(W_t)_{t \ge 0}$ is the Brownian motion on $\R^n$ with $W_0 = 0$.

We will show that if $X_t$ evolves following~\eqref{Eq:SDEex}, then its probability density function $\rho_t(x)$ evolves following the Fokker-Planck equation:
\begin{align}\label{Eq:FPex}
\part{\rho_t}{t} = -\nabla \cdot (\rho_t v) + \Delta \rho_t.
\end{align}
We can derive this heuristically as follows; we refer to standard textbooks for rigorous derivation~\cite{Mac92}.

For any smooth test function $\phi \colon \R^n \to \R$, let us compute the time derivative of the expectation
\begin{align*}
A(t) = \E_{\rho_t}[\phi] = \E[\phi(X_t)].
\end{align*}
On the one hand, we can compute this as
\begin{align}\label{Eq:FirstDot}
\dot A(t) = \frac{d}{dt} A(t) = \frac{d}{dt} \int_{\R^n} \rho_t(x) \phi(x) \, dx = \int_{\R^n} \part{\rho_t(x)}{t} \phi(x) \, dx.
\end{align}
On the other hand, by~\eqref{Eq:SDEex}, for small $\step > 0$ we have
\begin{align*}
X_{t+\step} &= X_t + \int_t^{t+\step} v(X_s) ds + \sqrt{2} (W_{t+\step}-W_t) \\
&= X_t + \step v(X_t) + \sqrt{2} (W_{t+\step}-W_t) + O(\step^2) \\
&\stackrel{d}{=} X_t + \step v(X_t) + \sqrt{2\step} Z + O(\step^2)
\end{align*}
where $Z \sim \N(0,I)$ is independent of $X_t$, since $W_{t+\step}-W_t \sim \N(0,\step I)$.
Then by Taylor expansion,
\begin{align*}
\phi(X_{t+\step}) &\stackrel{d}{=} \phi\left(X_t + \step v(X_t) + \sqrt{2\step} Z + O(\step^2)\right) \\
&= \phi(X_t) + \step \langle \nabla \phi(X_t), v(X_t) \rangle + \sqrt{2\step} \langle \nabla \phi(X_t), Z \rangle + \frac{1}{2} 2\step \langle Z, \nabla^2 \phi(X_t) Z \rangle + O(\step^{\frac{3}{2}}).
\end{align*}
Now we take expectation on both sides.
Since $Z \sim \N(0,I)$ is independent of $X_t$,
\begin{align*}
A(t+\step) &= \E[\phi(X_{t+\step})]\\
&= \E\left[\phi(X_t) + \step \langle \nabla \phi(X_t), v(X_t) \rangle + \sqrt{2\step} \langle \nabla \phi(X_t), Z \rangle + \step \langle Z, \nabla^2 \phi(X_t) Z \rangle\right] + O(\step^{\frac{3}{2}}) \\
&= A(t)+ \step \left(\E[\langle \nabla \phi(X_t), v(X_t) \rangle] + \E[\Delta \phi(X_t)]\right) + O(\step^{\frac{3}{2}}).
\end{align*}
Therefore, by integration by parts, this second approach gives
\begin{align}
\dot A(t) &= \lim_{ \step \to 0} \frac{A(t+\step)-A(t)}{\step} \notag \\
&= \E[\langle \nabla \phi(X_t), v(X_t) \rangle] + \E[\Delta \phi(X_t)] \notag \\
&= \int_{\R^n} \langle \nabla \phi(x), \rho_t(x) v(x) \rangle dx + \int_{\R^n} \rho_t(x) \Delta \phi(x) \, dx \notag \\
&= -\int_{\R^n} \phi(x) \nabla \cdot (\rho_t v)(x) \, dx + \int_{\R^n} \phi(x) \Delta \rho_t(x) \, dx \notag \\
&= \int_{\R^n} \phi(x) \left(-\nabla \cdot (\rho_t v)(x) + \Delta \rho_t(x)\right) \, dx.  \label{Eq:SecondDot}
\end{align}
Comparing~\eqref{Eq:FirstDot} and~\eqref{Eq:SecondDot}, and since $\phi$ is arbitrary, we conclude that
\begin{align*}
\part{\rho_t(x)}{t} = -\nabla \cdot (\rho_t v)(x) + \Delta \rho_t(x)
\end{align*}
as claimed in~\eqref{Eq:FPex}.

When $v = -\nabla f$, the stochastic differential equation~\eqref{Eq:SDEex} becomes the Langevin dynamics~\eqref{Eq:LD} from Section~\ref{Sec:Langevin}, and the Fokker-Planck equation~\eqref{Eq:FPex} becomes~\eqref{Eq:FP}.

In the proof of Lemma~\ref{Lem:OneStep}, we also apply the Fokker-Planck equation~\eqref{Eq:FPex} when $v = -\nabla f(x_0)$ is a constant vector field to derive the evolution equation~\eqref{Eq:FPCond} for one step of ULA.

\subsection{Remaining proofs}

\subsubsection{Proof of Lemma~\ref{Lem:LSILipschitz}}
\label{App:ProofLSILipschitz}

\begin{proof}[Proof of Lemma~\ref{Lem:LSILipschitz}]
Let $g \colon \R^n \to \R$ be a smooth function, and let $\tilde g \colon \R^n \to \R$ be the function $\tilde g(x) = g(T(x))$.
Let $X \sim \nu$, so $T(X) \sim \tilde \nu$.
Note that
\begin{subequations}
\begin{align*}
 \E_{\tilde \nu}[g^2] &=  \E_{X \sim \nu}[g(T(X))^2] = \E_{\nu}[\tilde g^2], \\
\E_{\tilde \nu}[g^2 \log g^2] &= \E_{X \sim \nu}[g(T(X))^2 \log g(T(X))^2] = \E_{\nu}[\tilde g^2 \log \tilde g^2].
 \end{align*}
\end{subequations}
Furthermore, we have $\nabla \tilde g(x) = \nabla T(x) \, \nabla g(T(x))$.
Since $T$ is $L$-Lipschitz, $\|\nabla T(x)\| \le L$.
Then
\begin{align*}
\|\nabla \tilde g(x)\| \le \|\nabla T(x)\| \, \|\nabla g(T(x))\| \le L \|\nabla g(T(x))\|.
\end{align*}
This implies
\begin{align*}
\E_{\tilde \nu}[\|\nabla g\|^2] = \E_{X \sim \nu}[\|\nabla g(T(X))\|^2] \ge \frac{\E_{\nu}[\|\nabla \tilde g\|^2]}{L^2}.
\end{align*}
Therefore,
\begin{align*}
\frac{\E_{\tilde \nu}[\|\nabla g\|^2]}{\E_{\tilde \nu}[g^2 \log g^2] - \E_{\tilde \nu}[g^2] \log \E_{\tilde \nu}[g^2]} 
&\ge \frac{1}{L^2} \, \frac{\E_{\nu}[\|\nabla \tilde g\|^2]}{\big(\E_{\nu}[\tilde g^2 \log \tilde g^2] - \E_{\nu}[\tilde g^2] \log \E_{\nu}[\tilde g^2]\big)}
\ge \frac{\alpha}{2L^2}
\end{align*}
where the last inequality follows from the assumption that $\nu$ satisfies LSI with constant $\alpha$.
This shows that $\tilde \nu$ satisfies LSI with constant $\alpha/L^2$, as desired.
\end{proof}

\subsubsection{Proof of Lemma~\ref{Lem:LSIGaussianConv}}
\label{App:ProofLSIGaussianConv}

\begin{proof}[Proof of Lemma~\ref{Lem:LSIGaussianConv}]
We recall the following convolution property of LSI~\cite{C04}:
If $\nu, \tilde \nu$ satisfy LSI  with constants $\alpha, \tilde \alpha > 0$, respectively, then $\nu \ast \tilde \nu$ satisfies LSI with constant $\left(\frac{1}{\alpha}+\frac{1}{\tilde\alpha}\right)^{-1}$.
Since $\N(0,2tI)$ satisfies LSI with constant $\frac{1}{2t}$, the claim follows.
\end{proof}

\subsubsection{Proof of Lemma~\ref{Lem:PLipschitz}}
\label{App:ProofPLipschitz}

\begin{proof}[Proof of Lemma~\ref{Lem:PLipschitz}]
Let $g \colon \R^n \to \R$ be a smooth function, and let $\tilde g \colon \R^n \to \R$ be the function $\tilde g(x) = g(T(x))$.
Let $X \sim \nu$, so $T(X) \sim \tilde \nu$.
Note that
\begin{align*}
\Var_{\tilde \nu}(g) &=  \Var_{X \sim \nu}(g(T(X))) = \Var_{\nu}(\tilde g).
 \end{align*}
Furthermore, we have $\nabla \tilde g(x) = \nabla T(x) \, \nabla g(T(x))$.
Since $T$ is $L$-Lipschitz, $\|\nabla T(x)\| \le L$.
Then
\begin{align*}
\|\nabla \tilde g(x)\| \le \|\nabla T(x)\| \, \|\nabla g(T(x))\| \le L \|\nabla g(T(x))\|.
\end{align*}
This implies
\begin{align*}
\E_{\tilde \nu}[\|\nabla g\|^2] = \E_{X \sim \nu}[\|\nabla g(T(X))\|^2] \ge \frac{\E_{\nu}[\|\nabla \tilde g\|^2]}{L^2}.
\end{align*}
Therefore,
\begin{align*}
\frac{\E_{\tilde \nu}[\|\nabla g\|^2]}{\Var_{\tilde \nu}(g)} 
&\ge \frac{1}{L^2} \, \frac{\E_{\nu}[\|\nabla \tilde g\|^2]}{\Var_{\nu}(\tilde g)}
\ge \frac{\alpha}{L^2}
\end{align*}
where the last inequality follows from the assumption that $\nu$ satisfies Poincar\'e inequality with constant $\alpha$.
This shows that $\tilde \nu$ satisfies Poincar\'e inequality with constant $\alpha/L^2$, as desired.
\end{proof}

\subsubsection{Proof of Lemma~\ref{Lem:PGaussianConv}}
\label{App:ProofPGaussianConv}

\begin{proof}[Proof of Lemma~\ref{Lem:PGaussianConv}]
We recall the following convolution property of Poincar\'e inequality~\cite{C18}:
If $\nu, \tilde \nu$ satisfy Poincar\'e inequality  with constants $\alpha, \tilde \alpha > 0$, respectively, then $\nu \ast \tilde \nu$ satisfies Poincar\'e inequality with constant $\left(\frac{1}{\alpha}+\frac{1}{\tilde\alpha}\right)^{-1}$.
Since $\N(0,2tI)$ satisfies Poincar\'e inequality with constant $\frac{1}{2t}$, the claim follows.
\end{proof}

\end{document}